\documentclass[lettersize,journal]{IEEEtran}
\usepackage{amsmath,amsfonts}
\usepackage{booktabs}
\usepackage{cite}
\usepackage{array}
\usepackage{tikz}
\usetikzlibrary{decorations.pathreplacing}
\usepackage{textcomp}
\usepackage{amsthm}
\usepackage[english]{babel}
\newtheorem{theorem}{Theorem}
\newtheorem{lemma}{Lemma}
\newtheorem{corollary}{Corollary}
\newtheorem{remark}{Remark}
\newtheorem{definition}{Definition}
\newtheorem{assumption}{Assumption}
\usepackage{stfloats}
\usepackage{graphicx}
\usepackage{xcolor}
\definecolor{my1}{RGB}{188, 100, 100}
\definecolor{my2}{RGB}{188, 188, 100}
\usepackage[caption=false,font=footnotesize]{subfig}
 
\usepackage{url}
\usepackage{verbatim}
\usepackage{graphicx}
\usepackage{amssymb}
\usepackage{algorithm2e}
\usepackage{amsmath}
\usepackage{algorithmic}
\usepackage{amsmath}

\renewcommand{\algorithmicrequire}{\textbf{Input:}}
\renewcommand{\algorithmicensure}{\textbf{Output:}}
\allowdisplaybreaks
\RestyleAlgo{ruled}
\def\BibTeX{{\rm B\kern-.05em{\sc i\kern-.025em b}\kern-.08em
    T\kern-.1667em\lower.7ex\hbox{E}\kern-.125emX}}
\usepackage{balance}
\usepackage{xcolor}
\usepackage{ifthen}
\usepackage{xcolor} 
\usepackage{soul}   

\usepackage{multirow} 
\definecolor{lightblue}{RGB}{173,216,230} 
\sethlcolor{lightblue}                    

\begin{document}
\title{Byzantine-Robust and Communication-Efficient Distributed Training: Compressive and Cyclic Gradient Coding}
\author{Chengxi Li,~\IEEEmembership{Member,~IEEE,}, Youssef Allouah, Rachid Guerraoui, Mikael Skoglund,~\IEEEmembership{Fellow, IEEE}, \\and Ming Xiao,~\IEEEmembership{Senior Member,~IEEE}
\thanks{The first author contributed the most to this work. The other authors are listed in alphabetical order by surname. C. Li, M. Xiao and M. Skoglund are with the Division of Information Science and Engineering, School of Electrical Engineering and Computer Science, KTH Royal Institute of Technology, 10044 Stockholm, Sweden. Youssef Allouah is with Stanford University, CA 94305, United States. Rachid Guerraoui is with School of Computer and Communication Sciences, the Swiss Federal Institute of Technology in Lausanne (EPFL), 1015 Lausanne, Switzerland. This work has been done when C. Li was a visiting postdoc researcher at EPFL under the MSCA Postdoctoral Fellowship 
 (e-mail: chengxli@kth.se; yallouah@stanford.edu; rachid.guerraoui@epfl.ch; skoglund@kth.se;mingx@kth.se;). Corresponding author: Chengxi Li.

This work was supported in part by MSCA Postdoctoral Fellowship under Grant 101146669, in part by Digital Futures, in part by Swedish Research Council under Grant 2021-04772 and Grant 2024-06464, and in part by the Swedish Innovation Agency (Vinnova) through the SweWIN
center (2023-00572).
}
}

\markboth{Journal of \LaTeX\ Class Files}%
{}

\maketitle

\begin{abstract}
In this paper, we study the problem of distributed training (DT) under Byzantine attacks with communication constraints. While prior work has developed various robust aggregation rules at the server to enhance robustness to Byzantine attacks, the existing methods suffer from a critical limitation in that the solution error does not diminish when the local gradients sent by different devices vary considerably, as a result of data heterogeneity among the subsets held by different devices. To overcome this limitation, we propose a novel DT method, cyc\textbf{l}ic gr\textbf{a}dient coding-based \textbf{D}T (LAD). In LAD, the server allocates the entire training dataset to the devices before training begins. In each iteration, it assigns computational tasks redundantly to the devices using cyclic gradient coding. Each honest device then computes local gradients on a fixed number of data subsets and encodes the local gradients before transmitting to the server. The server aggregates the coded vectors from the honest devices and the potentially incorrect messages from Byzantine devices using a robust aggregation rule. Leveraging the redundancy of computation across devices, the convergence performance of LAD is analytically characterized, demonstrating improved robustness against Byzantine attacks and significantly lower solution error.
Furthermore, we extend LAD to a communication-efficient variant, \textbf{com}pressive and cyc\textbf{l}ic gr\textbf{a}dient coding-based \textbf{D}T (Com-LAD), which further reduces communication overhead under constrained settings. Numerical results validate the effectiveness of the proposed methods in enhancing both Byzantine resilience and communication efficiency.
\end{abstract}


\section{Introduction}
\IEEEPARstart{D}{istributed} training (DT) has become a very important paradigm of modern machine learning, which enables the training of large-scale models by leveraging the computational resources of multiple devices \cite{chen2021distributed}. 
Under the typical setting of DT, before training begins, the subsets in the training data are allocated to the devices without overlap \cite{langer2020distributed}. During each iteration, the server sends the global model to the devices, and the devices compute local gradients corresponding to its local training data, which are then transmitted back to the server. After aggregating the local gradients from the devices, the server updates the global model accordingly.
In the above process, the allocation of training data to devices determines the distribution of computational tasks across them, since in each iteration all the training data assigned to a single device are used to compute its local gradient.
By distributing the computational tasks across numerous devices, DT systems can train machine learning models using massive datasets with reduced training time and better scalability to increasingly complex tasks\footnote{The motivation for DT is to leverage the computational resources of various devices. This is quite different from the motivation of the well-known federated learning scheme, in which the training data are originally distributed across devices and training must therefore be performed in a distributed manner.}. However, the distributed nature of DT systems leads to significant challenges.

First, DT systems are vulnerable to Byzantine attacks, where Byzantine devices may send incorrect or adversarial information to the server, significantly degrading training performance. To mitigate this threat, numerous DT methods have been proposed to enhance Byzantine robustness. For instance, in \cite{blanchard2017machine,xie2018phocas,xia2019faba,chen2017distributed,yin2018byzantine,pillutla2022robust,luan2024robust}, various robust aggregation rules have been developed to replace the traditional averaging-based aggregation at the server. While these methods achieve significantly better Byzantine resilience compared to averaging-based aggregation, they still suffer from severe performance degradation and a non-diminishing solution error in the presence of data heterogeneity among the subsets of the training data\footnote{It is a common case where the subsets are heterogeneous. For instance, if the training data are collected from different geographic regions, and each subset corresponds to data from a specific region, there are distinct characteristics among the subsets due to local variations.}. This limitation arises because robust aggregation rules typically assume that messages sent by honest devices are very close to each other, while messages from Byzantine devices deviate significantly. However, under significant data heterogeneity, the messages from honest devices may naturally vary widely, making it difficult to distinguish between honest and incorrect messages. As a result, Byzantine devices can more easily mislead the aggregation at the server and the overall training process, undermining the robustness of the system.

Another line of work involves using gradient coding techniques, where the subsets of the training data are allocated to the devices redundantly before training. Based on that, in \cite{hong2024group,hofmeister2024byzantine,data2020data,chen2018draco}, devices send coded gradients to the server in each iteration during the training process, and the server can identify the Byzantine devices with the received messages. In this way, the true global gradient can be applied to update the global model, which is the same as the case where there are no Byzantine attacks. 
Compared with existing methods that rely on robust aggregation rules, the use of gradient coding techniques can fully mitigate the negative impact of Byzantine devices. This allows the solution error to be reduced to zero, which is an outcome unattainable with robust aggregation rules alone. However, this level of Byzantine robustness comes at the expense of a significantly increased computational burden on the devices.

The second challenge is that the communication resources in DT systems are always limited, indicating that it is desirable to reduce the communication overhead caused by transmission from the devices to the server. There is lots of existing work aiming to deal with the communication bottleneck in DT. For example, quantization, sparsification, and other compression techniques have been widely studied to reduce the size of transmitted messages from the devices \cite{jiang2018linear,shi2019understanding,wangni2018gradient,beznosikov2023biased,xie2024jointsq}. These methods typically focus on reducing the communication cost in settings without Byzantine attacks, where all devices are cooperative and honest. However, in the presence of Byzantine devices, directly applying these methods results in significantly degraded learning performance.

In the DT literature, Byzantine robustness and communication efficiency have been investigated as individual challenges. There is little work that addresses these two challenges simultaneously by combining the advantages of robust aggregation rules with communication compression techniques \cite{ghosh2021communication,zhu2023byzantine,rammal2024communication,gorbunov2022variance}. However, like prior work that designs robust aggregation rules at the server, the performance of these existing methods is inherently limited by data heterogeneity among subsets. In other words, in the presence of data heterogeneity, the messages sent by honest devices in the compressed domain differ significantly from one another, leading to a non-diminishing solution error and degraded training performance. 

\textbf{Contributions.} To overcome the limitations of the existing methods, in this paper, we propose a novel DT method, cyc\textbf{l}ic gr\textbf{a}dient coding-based \textbf{D}T (LAD). In LAD, before the training, the server allocates the entire training dataset, which consists of a number of subsets, to all devices. During the training, in each iteration, the server assigns computational tasks redundantly to the devices using cyclic gradient coding. Based on that, each device computes and encodes the local gradients corresponding to a fixed number of data subsets from the entire training dataset. For the honest devices, the coded vectors are sent to the server subsequently. The server aggregates the trustworthy vectors from the honest devices and the incorrect messages from the Byzantine devices using a robust aggregation rule to update the global model. Note that LAD is a meta-algorithm in which a vast number of existing robust aggregation rules can be incorporated and applied.
Furthermore, LAD is extended to a communication-efficient variant, \textbf{com}pressive and cyc\textbf{l}ic gr\textbf{a}dient coding-based \textbf{D}T (Com-LAD), where the communication overhead is further reduced by letting the devices to compress the coded vectors before the transmission. For the proposed methods, the convergence performance is analyzed and presented, which shows that the proposed methods attain improved robustness against Byzantine attacks and significantly lower solution error, compared with the methods that adopt the robust aggregation rules alone. 
Finally, we present various numerical results to demonstrate the superiority of the proposed methods in improving the Byzantine robustness and communication efficiency.

\section{Related Work}
\subsection{Byzantine-Robust DT}
To mitigate the threat of Byzantine attacks, numerous DT methods have been developed to enhance Byzantine robustness. 

\textbf{Robust aggregation rules.} Some of the existing methods aim to design robust aggregation rules to replace the averaging-based aggregation rule at the server under the typical setting.
For example, in \cite{blanchard2017machine}, the server aggregates the information sent by the devices using majority-based and squared-distance-based methods, selecting trustworthy vectors as those that minimize the distances to their neighbors. In \cite{xie2018phocas} and \cite{yin2018byzantine}, robust DT methods are developed based on trimmed mean and median aggregation rules at the server.
In \cite{xia2019faba}, outliers among the messages received by the server are removed, and only local gradients closer to the true gradients are retained. In \cite{chen2017distributed} and \cite{pillutla2022robust}, aggregation at the server is performed using the geometric median, ensuring that the individual contribution of each device is not revealed.
In \cite{luan2024robust}, a robust aggregation rule based on the maximum correntropy criterion is utilized, where the correntropy function serves as the similarity measure. In \cite{allouah2023fixing}, the authors propose an aggregation method based on nearest neighbor mixing, which achieves optimal Byzantine resilience under data heterogeneity. 
These robust aggregation rules offer improved Byzantine resilience compared to simple averaging. However, under data heterogeneity, when different subsets of the training dataset are not independent and identically distributed (non-IID), they suffer from a non-diminishing solution error, which increases significantly as the level of heterogeneity grows.

\textbf{Gradient coding.} There is also another line of work dealing with Byzantine attacks by using gradient coding techniques, where the subsets of the training data are allocated to the devices redundantly before training and this redundancy is utilized during the training to enhance Byzantine robustness. 
For instance, for matrix multiplication tasks in DT, in \cite{hong2024group}, coding techniques are integrated with group-wise verification to counter Byzantine attacks. For the same problem, \cite{data2020data} proposes an error-correction-based approach for matrix-vector multiplication under Byzantine attacks, which is information-theoretically optimal.  
For DT problems in a more general setting, in \cite{chen2018draco}, local gradients are encoded using fractional and cyclic repetition codes, with decoding performed via majority vote and Fourier techniques. \cite{hofmeister2024byzantine} employs fractional repetition codes to allocate the training subsets. Based on that, it can detect and transform messages from Byzantine devices into erasures with additional computational burden at the server and larger communication overhead during the training.

In contrast to robust aggregation rules, which result in a non-diminishing solution error under data heterogeneity, gradient coding techniques can reduce the solution error to zero. However, this enhanced Byzantine robustness comes at the expense of a substantially increased computational burden on the devices.

\subsection{Communication-Efficient DT}
When all devices are cooperative and honest, existing research proposes various communication compression techniques to cope with the communication bottleneck in DT. For example, in \cite{jiang2018linear}, an unbiased stochastic quantization function is adopted as the communication compression function, and its convergence rate is analyzed. In \cite{shi2019understanding,wangni2018gradient}, sparsified gradient vectors are transmitted by devices to reduce communication overhead. In \cite{beznosikov2023biased}, three classes of biased communication compression functions are developed, and a linear convergence rate is demonstrated. In \cite{xie2024jointsq}, a joint Sparsification-Quantization scheme is proposed for DT, which can be regarded as a mixed-precision quantization process. These methods are designed to reduce communication overhead in settings without Byzantine attacks, under the assumption that all devices are cooperative and honest.

\subsection{Byzantine-Robust and Communication-Efficient DT}

Byzantine robustness and communication efficiency in DT have been studied as individual yet significant topics. However, there is little work that simultaneously explores these two challenges in DT. The few existing studies addressing both are as follows. In \cite{ghosh2021communication}, a DT method is proposed that performs thresholding based on gradient norms at the server to remove messages from Byzantine devices, incorporating a generic class of 
$\delta$-approximate compression functions to enhance communication efficiency. In \cite{zhu2023byzantine}, a DT method is introduced that combines the advantages of communication compression techniques and the geometric median aggregation rule. In \cite{rammal2024communication} and \cite{gorbunov2022variance}, Byzantine-resilient DT methods are developed based on variance reduction with compressed communication.

However, when applying these methods, the devices transmit compressed messages derived from local gradients, and the server aggregates these messages using robust aggregation rules. Due to the inherent limitations of robust aggregation based on local gradient information, the learning performance of these methods is constrained by data heterogeneity among subsets, regardless of whether the aggregation is implemented in the compressed domain or not.

\section{Problem Model}
\label{problem model}
\subsection{DT Systems}
We consider a DT system, where there are $N$ devices and a central server with a training dataset $\mathcal{D}$. Here, the training dataset $\mathcal{D}$ is divided into $N$ subsets, represented as $\mathcal{D} = \{ \mathcal{D}_1, \dots, \mathcal{D}_N \}$. The goal is to train a model using the training dataset, which is equivalent to solving the following optimization problem:
\begin{align}
    \label{basic problem}
   {{\mathbf{x}}^*} = \arg {\min _{{\mathbf{x}} \in {\mathbb{R}^{Q\times1}}}}F\left( {\mathbf{x}} \right).
\end{align}
Here, ${\mathbf{x}}$ denotes the parameter vector in the trained model, and $F\left( {\mathbf{x}} \right)$ represents the training loss on the entire training dataset, defined as
\begin{align}
\label{overall loss}
F({\mathbf{x}}) = \sum\limits_{\varrho  \in \mathcal{D}} {l\left( {{\mathbf{x}},\varrho } \right)} ,
\end{align}
where \({l}\left( {{\mathbf{x}},{\varrho}} \right): \mathbb{R}^{Q\times1} \to \mathbb{R}\) is the training loss on data sample \({{\varrho}}\) in \({\mathcal{D}}\). 

Note that, we can equivalently express the optimization problem in (\ref{basic problem}) as:
\begin{align}
    \label{basic problem 2}
   {{\mathbf{x}}^*} = \arg {\min _{{\mathbf{x}} \in {\mathbb{R}^{Q\times1}}}}F\left( {\mathbf{x}} \right) = \arg {\min _{{\mathbf{x}} \in {\mathbb{R}^{Q\times1}}}}\sum\limits_{k = 1}^N {{f_k}\left( {\mathbf{x}} \right)}, 
\end{align}
where \(f_k({\mathbf{x}}): \mathbb{R}^{Q\times1} \to \mathbb{R}\) is the training loss on \({\mathcal{D}_k}\), given as 
\begin{align}
\label{fix}
{f_k}({\mathbf{x}}) = \sum\limits_{\varrho  \in {\mathcal{D}_k}} {l\left( {{\mathbf{x}},\varrho } \right)}.
\end{align}

To exploit the computational resources of the $N$ devices in the system, the subsets in $\mathcal{D}$ are allocated to the devices in a certain manner before the training \cite{langer2020distributed}. In training iteration $t$, the server broadcasts the global model $\mathbf{x}^t$ to all devices. Upon receiving the global model, each device computes the local gradient based on its local training data and transmits this gradient back to the server. The server aggregates the local gradients received from the devices to obtain the global gradient and updates the global model accordingly. 

\subsection{Byzantine Attacks}
In many practical scenarios, the DT system may be subject to Byzantine attacks, causing some devices to act as Byzantine devices \cite{guerraoui2024byzantine,allouah2023fixing}. In training iteration $t$, let $\mathcal{B}^t$ denote the set of indices of all Byzantine devices, and let $\mathcal{H}^t$ denote the set of indices of all honest devices. The sets $\mathcal{B}^t$ and $\mathcal{H}^t$ may remain fixed or vary across training iterations, and they are unknown to the server. During each iteration, the messages sent by honest devices are completely trustworthy, whereas those from Byzantine devices can be arbitrarily incorrect and misleading. Note that no further assumptions are required, provided that the messages sent by the Byzantine devices have finite norms. Here, it is assumed that, in each iteration, at most $N - H$ devices are Byzantine, where $H > \frac{N}{2}$. In this case, it is sufficient to consider the scenario where exactly $N - H$ devices are Byzantine, since the method and performance bound for this worst-case setting also apply when fewer than $N - H$ devices are Byzantine.

\subsection{Communication Bottleneck}
In the DT system, communication resources are always limited, such as bandwidth and network capacity \cite{beznosikov2023biased,xie2024jointsq}. To address this communication bottleneck, it is highly desirable to reduce the overhead caused by devices transmitting various messages to the server in each iteration. 

\subsection{Objective}
Based on the above descriptions of the DT system, the objective of the system designer is to develop DT methods that achieve better learning performance, even in the presence of Byzantine attacks, under certain communication constraints. In other words, the goal is to enhance the Byzantine resilience and communication efficiency of the DT system.

\section{Proposed Method without Compressed Communication}
\label{our method without compression}

In this section, we first introduce the proposed method based on cyc\textbf{l}ic gr\textbf{a}dient co\textbf{d}ing (LAD), to deal with the DT problem formulated in Section~\ref{problem model} under Byzantine attacks without considering communication bottlenecks. Based on LAD, the proposed method with compressed communication, which enhances the Byzantine resilience and communication efficiency simultaneously, will be introduced in Section~\ref{our method with compression}. 

Before training, all $N$ subsets are allocated to each device. In other words, each device obtains a copy of the entire training dataset. Based on this, let us define a computation task matrix $\mathbf{\hat{S}}$ of size $N \times N$, which is a cyclic matrix where each row is a cyclic shift of the previous one and $\hat s(i,k)$ denotes the $(i,k)$-th element. Without loss of generality, in the first row of $\mathbf{\hat{S}}$, the first $d$ elements are ones, while the remaining elements are zeros. In $\mathbf{\hat{S}}$, each row represents a computation task. For the $i$-th row, the computation task involves computations based on the $d$ subsets, which will explained in more details later. And we will show in Corollary~\ref{coro} in Section~\ref{convergence performance} that, by setting the computation task matrix as $\mathbf{\hat{S}}$, better robustness to Byzantine attacks can be attained.  

During training, in iteration $t$, the server broadcasts the current global model $\mathbf{x}^t$ to the devices. Simultaneously, the server generates a list of task indices ${\mathcal{T}_1^t, \dots, \mathcal{T}_N^t}$ as a random permutation of ${1, \dots, N}$, and then transmits the task index $\mathcal{T}_i^t$ to device $i$, $\forall i$. The server also generates another random permutation of ${1, \dots, N}$ independently from the generation of the task indices, denoted as \({{\mathbf{p}}^t} = \left[ {p_1^t,...,p_N^t} \right]\), and broadcasts $\mathbf{p}^t$ to all devices. Note that, both the generation of the task indices and the generation of the vector ${{\mathbf{p}}^t}$ are independent across iterations.   

Upon receiving the global model, device $i$ computes the local gradients based on the received task index $\mathcal{T}_i^t$ and ${{\mathbf{p}}^t}$, computing only the local gradients $\{ \nabla {f_{p_k^t}}({{\mathbf{x}}^t})|\hat s(\mathcal{T}_i^t,k) = 1\}$. Based on this, $d$ represents the computational load on each device in a single iteration, where a larger value of $d$ implies a higher computational load.

\begin{remark}
    Note that in LAD, the computation tasks assigned to devices are redundant, as the local gradient corresponding to any subset can be computed by multiple devices. This setup resembles the scenario in conventional gradient coding techniques, which were originally proposed to address the straggler problem in DT. However, the motivations are distinct. In conventional gradient coding, redundant computations are exploited to compensate for the missing information from stragglers by leveraging the results from non-stragglers. In contrast, in LAD, redundancy is used to reduce the variance among messages sent by honest devices, thereby enhancing robustness against Byzantine attacks, which will be shown in Section~\ref{performance analysis}. 
\end{remark}

Subsequently, device $i$ encodes the local gradients into a coded vector, given as
\begin{align}
    \label{encoding}
    {\mathbf{g}}_i^t = \sum\limits_{k \in \left\{ {\left. k \right|\hat s\left( {\mathcal{T}_i^t,k} \right) = 1} \right\}} {\frac{1}{d}} \nabla {f_{p_k^t}}\left( {{{\mathbf{x}}^t}} \right).
\end{align}

If device $i$ is an honest device, where $i \in \mathcal{H}^t$, it sends ${\mathbf{ g}}_i^t$ to the server. Otherwise, if $i \in \mathcal{B}^t$, device $i$ transmits $\mathbf{b}_i^t \in \mathbb{R}^{Q \times 1}$ to the server, where $\mathbf{b}_i^t$ is an incorrect message compared to ${\mathbf{g}}_i^t$.

Upon receiving messages from all devices, the server employs a robust aggregation rule, denoted as $agg(\cdot)$, to obtain the global model update. Considering that $\kappa$-robustness provides a unified and sufficiently fine-grained characterization of various robustness criteria, the server employs a robust aggregation rule with $\kappa$-robustness. As in \cite{allouah2023fixing}, the definition of $\kappa$-robustness is given as follows. 

\begin{definition}[{\bf $\kappa$-robustness}]
\label{def:resaveraging}
Suppose $\frac{N-H}{N} < \frac{1}{2}$ and $\kappa \geq 0$, where $N$ and $H$ are two integers with $N>H$. An aggregation rule $\text{agg}(\cdot)$ is said to satisfy $\kappa$-robustness if, for any $H$ vectors $\{\mathbf{z}_1, \dots, \mathbf{z}_{H}\}$ and $N-H$ vectors $\{\tilde{\mathbf{z}}_1, \dots, \tilde{\mathbf{z}}_{N-H}\}$ of finite norms, all in $\mathbb{R}^{Q \times 1}$, the following holds:
\begin{align*}
    \left\| \text{agg}\left( \left\{ \mathbf{z}_i \right\}_{i = 1}^{H}, \left\{ \tilde{\mathbf{z}}_j \right\}_{j = 1}^{N-H} \right) - \bar{\mathbf{z}} \right\|^2 
    \leq \kappa \cdot \frac{1}{H} \sum_{i = 1}^{H} \left\| \mathbf{z}_i - \bar{\mathbf{z}} \right\|^2,
\end{align*}
where $\bar{\mathbf{z}} = \frac{1}{H} \sum_{i=1}^{H} \mathbf{z}_i$. Here, $\kappa$ is the \emph{robustness coefficient}.
\end{definition}

Based on that, the aggregation at the server can be expressed as
\begin{align}
    \label{global model update wout com}
    {{\mathbf{g}}^t} = agg\left( {{{\left\{ {{\mathbf{g}}_i^t} \right\}}_{i \in {\mathcal{H}^t}}},{{\left\{ {{\mathbf{b}}_j^t} \right\}}_{j \in {\mathcal{B}^t}}}} \right).
\end{align}
At the end of iteration $t$, the server updates the global model as 
\begin{align}
    \label{update global model wout com}
    {{\mathbf{x}}^{t + 1}} = {{\mathbf{x}}^t} - {\gamma ^t}{{\mathbf{g}}^t},
\end{align}
where $\gamma ^t$ denotes the learning rate. The proposed method is presented in Algorithm 1, and the implementation of a single iteration is shown as Fig.~\ref{fig: method}.
\begin{figure*}
    \centering
    \includegraphics[width=0.8\linewidth]{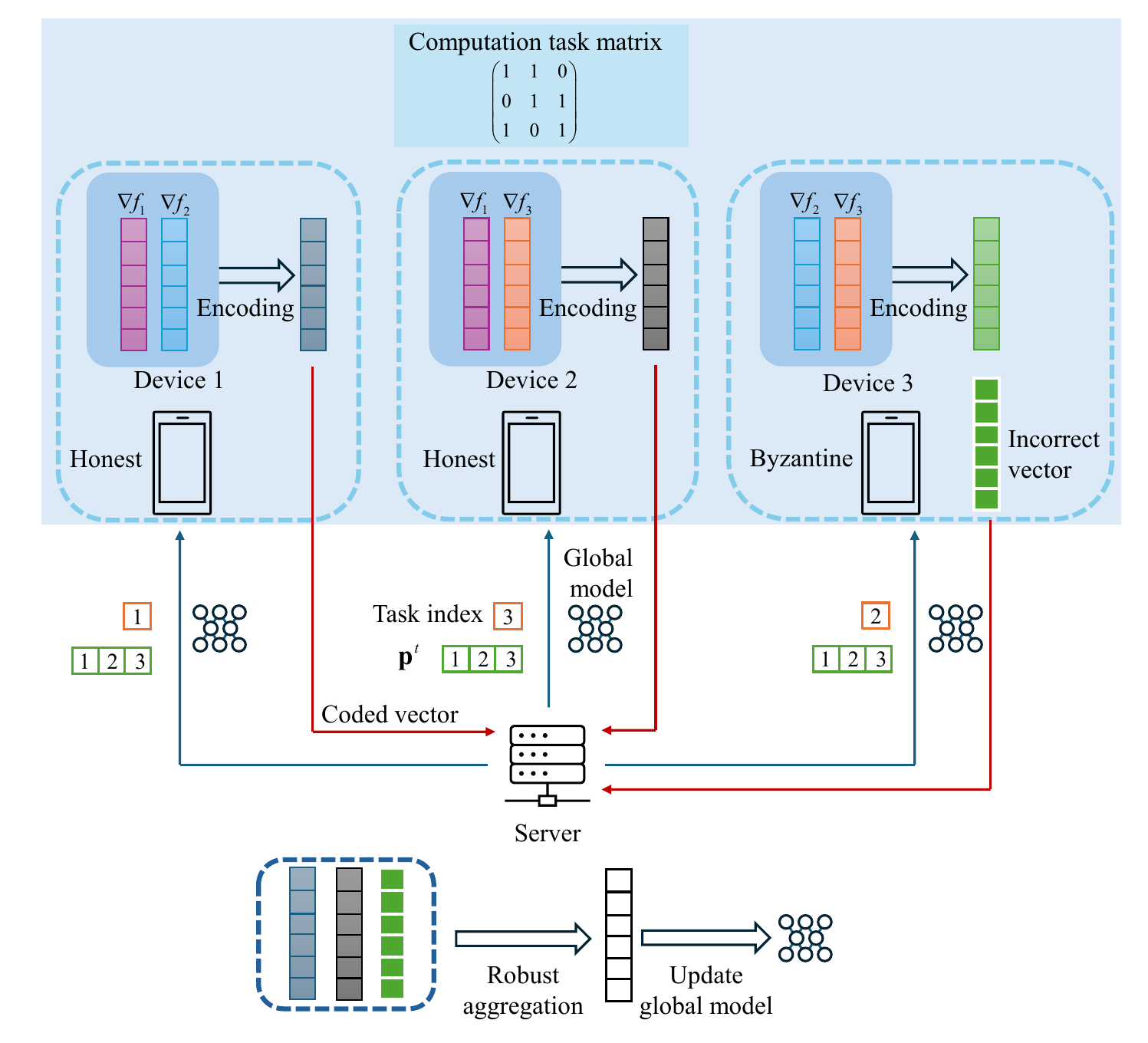}
    \caption{The implementation of a single iteration of LAD.}
    \label{fig: method}
\end{figure*}

\begin{algorithm}
\label{alg1}
\caption{LAD}
\KwIn{Training subsets $\mathcal{D} = \{\mathcal{D}_1, \dots, \mathcal{D}_N\}$, learning rate $\{\gamma^t\}$, computational load $d$, number of iterations $T$.}
\KwOut{Global model $\mathbf{x}^T$.}

\textbf{Initialization:} Initialize global model $\mathbf{x}^0$.

\For{$t = 0$ \KwTo $T-1$}{
    Server broadcasts current global model $\mathbf{x}^t$ to all devices.\\
    Generate a list of task indices $\{\mathcal{T}_1^t,\dots,\mathcal{T}_N^t\}$ as a random permutation of $\{1,\dots,N\}$.\\
    Generate ${{\mathbf{p}}^t}$ as a random permutation of $\{1,\dots,N\}$.\\
    Transmit task index $\mathcal{T}_i^t \in \{1,\dots,N\}$ to device $i$, $\forall i$.\\
    Broadcast ${{\mathbf{p}}^t}$ to all devices. 

    \For{each device $i$ in parallel}{
        Compute local gradients: $\{ \nabla {f_{p_k^t}}({{\mathbf{x}}^t})|\hat s(\mathcal{T}_i^t,k) = 1\}$. \\
        Encode local gradients into coded vector ${\mathbf{g}}_i^t$ as (\ref{encoding}). \\
        \eIf{$i \in \mathcal{H}^t$ (honest)}{
            Send $\mathbf{{g}}_i^t$ to the server.
        }{
            Send incorrect vector $\mathbf{{b}}_i^t$ to the server.
        }
    }

    Server receives $\{\mathbf{{g}}_i^t\}_{i \in \mathcal{H}^t}$ and $\{\mathbf{{b}}_j^t\}_{j \in \mathcal{B}^t}$. \\
    Aggregate received messages using a robust aggregation rule as (\ref{global model update wout com}) to obtain ${{\mathbf{g}}^t}$. \\
    Update global model as (\ref{update global model wout com}).
}

\Return{$\mathbf{x}^T$}
\end{algorithm}

\section{Proposed Method with Compressed Communication}
\label{our method with compression}

In this section, we extend the proposed method in Section~\ref{our method without compression} and propose a DT method based on \textbf{com}pressive and cyc\textbf{l}ic gr\textbf{a}dient co\textbf{d}ing (Com-LAD), to deal with the problem formulated in Section~\ref{problem model} under both challenges of Byzantine attacks and communication bottlenecks. 

Before training, the subsets are allocated to each device, as in LAD. The same computation task matrix $\mathbf{\hat{S}}$ is defined and used as in LAD. During training, in iteration $t$, the server sends the global model $\mathbf{x}^t$ and the vector $\mathbf{p}^t$ to the devices and transmits the random task index $\mathcal{T}_i^t$ to device $i$ for all $i$, following the same procedure as in LAD. Upon receiving the global model, device $i$ computes the local gradients and encodes them as in (\ref{encoding}).

For device $i$, to reduce the communication overhead caused by uploading to the server, before transmission to the server, it employs a compression function $\mathcal{C}: \mathbb{R}^{Q\times1} \to \mathbb{R}^{Q\times1}$ to compress ${\mathbf{g}}_i^t$ into ${\mathbf{\hat g}}_i^t$, where
\begin{align}
    \label{compress}
    {\mathbf{\hat g}}_i^t = \mathcal{C}\left( {\mathbf{g}}_i^t \right).
\end{align}
There are various commonly used compression functions, which can be either biased or unbiased. In this work, we adopt unbiased compression functions, as they generally achieve more stable convergence under typical DT settings. The definition of unbiased compression functions is given as follows \cite{he2023unbiased,condat2022ef}.

\begin{definition}[Unbiased compression functions]
\label{def_unbiased_com}
A compression function $\mathcal{C}: \mathbb{R}^{Q\times1} \to \mathbb{R}^{Q\times1}$ is an unbiased compression function, if it holds that
\begin{align}
    \label{def unbiased}
    \mathbb{E}\left[ {\mathcal{C}\left( {\mathbf{g}} \right)} \right] &= {\mathbf{g}},\\
    \label{def unbiased bound}
    \mathbb{E}\left[ {{{\left\| {\mathcal{C}\left( {\mathbf{g}} \right) - {\mathbf{g}}} \right\|}^2}} \right] &\leqslant \delta {\left\| {\mathbf{g}} \right\|^2}, \forall \mathbf{g}.
\end{align}
\end{definition}
Here, in (\ref{def unbiased bound}), $\delta$ is a non-negative constant. As a special case, when $\delta = 0$, there is no communication compression. As $\delta$ increases, the compression error becomes larger.

Note that although the length of the vectors as the input and output of the compression function is the same, the communication overhead caused by transmitting $\mathcal{C}(\mathbf{g})$ is reduced compared to that of transmitting $\mathbf{g}$, since the former requires fewer bits.
Some commonly used unbiased compression functions are listed as follows:
\begin{itemize}
    \item Stochastic quantization \cite{alistarh2017qsgd}. For each element in $\mathbf{g}$, denoted by $g_q \in [a, b]$ for $q = 1, \dots, Q$, it is quantized to $a$ with probability $\frac{b - g_q}{b - a}$ and to $b$ with probability $\frac{g_q - a}{b - a}$.
    \item Random sparsification \cite{wangni2018gradient}. For any $\mathbf{g}$, randomly select $\hat{Q}$ elements and scale them by a factor of $\frac{Q}{\hat{Q}}$, while setting the remaining elements to zero. 
\end{itemize}

If device $i$ is honest, it sends ${\mathbf{\hat g}}_i^t$ to the server. Otherwise, device $i$ transmits $\mathbf{\hat b}_i^t \in \mathbb{R}^{Q \times 1}$ to the server. Here, $\mathbf{\hat b}_i^t$ is an incorrect message compared to ${\mathbf{\hat g}}_i^t$. 
Upon receiving messages from all devices, the server employs $agg(\cdot)$ with $\kappa$-robustness to generate the global model update as  
\begin{align}
    \label{global model update}
    {{\mathbf{\hat g}}^t} = agg\left( {{{\left\{ {{\mathbf{\hat g}}_i^t} \right\}}_{i \in {\mathcal{H}^t}}},{{\left\{ {{\mathbf{\hat b}}_j^t} \right\}}_{j \in {\mathcal{B}^t}}}} \right).
\end{align}
At the end of iteration $t$, global model is updated at the server: 
\begin{align}
    \label{update global model}
    {{\mathbf{x}}^{t + 1}} = {{\mathbf{x}}^t} - {\gamma ^t}{{\mathbf{\hat g}}^t},
\end{align}
where $\gamma ^t$ is the learning rate. We present the proposed method in Algorithm 2. 

\begin{algorithm}
\label{alg2}
\caption{Com-LAD}
\KwIn{Training subsets $\mathcal{D} = \{\mathcal{D}_1, \dots, \mathcal{D}_N\}$, learning rate $\{\gamma^t\}$, computational load $d$, number of iterations $T$.}
\KwOut{Global model $\mathbf{x}^T$.}

\textbf{Initialization:} Initialize global model $\mathbf{x}^0$.

\For{$t = 0$ \KwTo $T-1$}{
    Server broadcasts current global model $\mathbf{x}^t$ to all devices.\\
    Generate a list of task indices $\{\mathcal{T}_1^t,\dots,\mathcal{T}_N^t\}$ as a random permutation of $\{1,\dots,N\}$.\\
    Generate ${{\mathbf{p}}^t}$ as a random permutation of $\{1,\dots,N\}$.\\
    Transmit task index $\mathcal{T}_i^t \in \{1,\dots,N\}$ to device $i$, $\forall i$.\\
    Broadcast ${{\mathbf{p}}^t}$ to all devices. 

    \For{each device $i$ in parallel}{
        Compute local gradients: $\{ \nabla {f_{p_k^t}}({{\mathbf{x}}^t})|\hat s(\mathcal{T}_i^t,k) = 1\}$. \\
        Encode local gradients into coded vector ${\mathbf{g}}_i^t$ as (\ref{encoding}). \\
        Compress coded vector as ${\mathbf{\hat g}}_i^t = \mathcal{C}\left( {\mathbf{g}}_i^t \right)$. \\
        \eIf{$i \in \mathcal{H}^t$ (honest)}{
            Send $\mathbf{\hat{g}}_i^t$ to the server.
        }{
            Send incorrect vector $\mathbf{\hat{b}}_i^t$ to the server.
        }
    }

    Server receives $\{\mathbf{\hat{g}}_i^t\}_{i \in \mathcal{H}^t}$ and $\{\mathbf{\hat{b}}_j^t\}_{j \in \mathcal{B}^t}$. \\
    Aggregate received messages using a robust aggregation rule as (\ref{global model update}) to obtain ${{\mathbf{\hat g}}^t}$. \\
    Update global model as (\ref{update global model}).
}

\Return{$\mathbf{x}^T$}
\end{algorithm}

\section{Convergence Analysis}
\label{performance analysis}
In this section, we analyze the convergence performance of LAD and Com-LAD. To this end, some assumptions are made, which have been widely used under various DT settings.  
\begin{assumption}
    \label{smooth assumption}
The training loss function $F$ is an $L$-smooth function \cite{jadbabaie2023federated,allouah2024fine}, implying  
    \begin{align}
        \label{smooth assum}
  F\left( {\mathbf{x}} \right) \leqslant F\left( {\mathbf{y}} \right) + \left\langle {\nabla F\left( {\mathbf{y}} \right),{\mathbf{x}} - {\mathbf{y}}} \right\rangle  + \frac{L}{2}{\left\| {{\mathbf{x}} - {\mathbf{y}}} \right\|^2},\forall {\mathbf{x}},{\mathbf{y}},
  \end{align}
where $L>0$ is a constant. 
\end{assumption}
\begin{assumption}
    \label{assp bounded heter}
    It holds that \cite{dong2023byzantine}
    \begin{align}
        \label{bounded heterogenity}
        \frac{1}{N}\sum\limits_{i = 1}^N {{{\left\| {\nabla {f_i}\left( {\mathbf{x}} \right) - {{\boldsymbol{\mu }}^t}} \right\|}^2}} \leqslant {\beta ^2},\forall {\mathbf{x}},
    \end{align}
where 
\begin{align}
    \label{def mu}
  \boldsymbol{\mu}^t = \frac{1}{N}\nabla F\left( {{{\mathbf{x}}^t}} \right).
\end{align}
This indicates that the heterogeneity among $\left\{ {{\mathcal{D}_1},...,{\mathcal{D}_N}} \right\}$ is upper bounded. 
\end{assumption}
\begin{assumption}
    \label{lower bound}
    There exists a constant ${{F^*}}$ acting as the lower bound of the training loss function \cite{horvath2023stochastic}, where
    \begin{align}
        \label{lower bound of F}
        F\left( {\mathbf{x}} \right) \geqslant {F^*},\forall \mathbf{x}.
    \end{align}
\end{assumption}
Based on these assumptions, we provide several lemmas to support the derivation of the main theorems that characterize the convergence performance of LAD and Com-LAD. 
\begin{lemma}
\label{lemma1}
Suppose the set $\mathcal{S}$ contains all matrices in which each row has exactly $d$ entries equal to one, and all other entries are zero. It holds that 
\begin{align}
    \label{S bound}
   & \mathbb{E}\left( {{{\left\| {\left( {\frac{1}{d}\frac{1}{{H}}{{\mathbf{h}}^{1 \times N}}{\mathbf{\hat S}} - \frac{1}{N}{{\mathbf{1}}^{1 \times N}}} \right)} \right\|}^2}} \right)\nonumber \\
    = &{\inf _{{\mathbf{S}} \in \mathcal{S}}}\mathbb{E}\left( {{{\left\| {\left( {\frac{1}{d}\frac{1}{{H}}{{\mathbf{h}}^{1 \times N}}{\mathbf{S}} - \frac{1}{N}{{\mathbf{1}}^{1 \times N}}} \right)} \right\|}^2}} \right)\nonumber\\
    =&\frac{{\left( {N - H} \right)\left( {N - d} \right)}}{{dH\left( {N - 1} \right)N}}. 
\end{align}
Here, $\mathbf{h}^{1 \times N}$ is a random vector where exactly $H$ elements are randomly and uniformly selected to be ones, and the remaining elements are zeros.
\end{lemma}
\begin{proof}
    Please see Appendix~\ref{appen lemma1} in the supplementary materials. 
\end{proof}
\begin{lemma}
    \label{lemma20}
    We can obtain
\begin{align}
    \label{h1}
    \mathbb{E}\left( {\left. {{{\left\| {{\mathbf{g}}_i^t - {{\boldsymbol{\mu }}^t}} \right\|}^2}} \right|{\mathcal{F}^t}} \right) \leqslant \frac{{\left( {N - d} \right){\beta ^2}}}{{d\left( {N - 1} \right)}}.
\end{align}
\end{lemma}
\begin{proof}
    Please see Appendix~\ref{appen lemma20} in the supplementary materials. 
\end{proof}
\begin{lemma}
    \label{lemma2}
Let us define
\begin{align}
    \label{honest avg}
    {{\mathbf{\bar g}}^t} = \frac{1}{H}\sum\limits_{i \in {\mathcal{H}^t}} {{\mathbf{\hat{g}}}_i^t}.
\end{align}
Based on (\ref{honest avg}), we have
    \begin{align}
        \label{bound honest avg}
        \mathbb{E}\left( {\left. {{{\left\| {{{{\mathbf{\hat g}}}^t} - {{{\mathbf{\bar g}}}^t}} \right\|}^2}} \right|{\mathcal{F}^t}} \right) \leqslant \kappa {\kappa _1} + \kappa {\kappa _2}{\left\| {\nabla F\left( {{{\mathbf{x}}^t}} \right)} \right\|^2},
    \end{align}
where $\mathbb{E}\left( {\left. {\cdot} \right|{\mathcal{F}^t}} \right)$ is the expectation conditioned on previous iterations and
\begin{align}
    \label{def k1}
    {\kappa _1} \triangleq N{\beta ^2}\left[ {\left( {\frac{1}{H} + 1} \right)\frac{{4\delta }}{d}} \right] + 4{\beta ^2}\frac{{\left( {N - d} \right)N}}{{dH\left( {N - 1} \right)}},
\end{align}
\begin{align}
    \label{def k2}
    {\kappa _2} \triangleq \left[ {\left( {\frac{1}{H} + 1} \right)\frac{{4\delta }}{d} + 4\frac{{\left( {N - H} \right)\left( {N - d} \right)}}{{dH\left( {N - 1} \right)N}}} \right]\frac{1}{N}.
\end{align}
\end{lemma}
\begin{proof}
    Please see Appendix~\ref{appen lemma2} in the supplementary materials, where Lemma~\ref{lemma1} and Lemma~\ref{lemma20} are applied. 
\end{proof}
\begin{corollary}
\label{coro}
By setting the computation task matrix as $\mathbf{\hat{S}}$, the infimum of the upper bound of $\mathbb{E}\left( {\left. {{{\left\| {{{{\mathbf{\hat g}}}^t} - {{{\mathbf{\bar g}}}^t}} \right\|}^2}} \right|{\mathcal{F}^t}} \right)$ can be attained. In that case, the deviation between the global model update and the average of the messages from the honest devices is minimized, which contributes to better robustness to Byzantine attacks. 
\end{corollary}
\begin{proof}
Please see Appendix~\ref{appen coro1} in the supplementary materials. 
\end{proof}
\begin{lemma}
    \label{lemma4}
For the average of the messages sent by the honest devices, we can derive 
\begin{align}
    \label{g bar}
    \mathbb{E}\left( {\left. {{{\left\| {{{{\mathbf{\bar g}}}^t}} \right\|}^2}} \right|{\mathcal{F}^t}} \right) \leqslant {\kappa _3} + {\kappa _4}{\left\| {\nabla F\left( {{{\mathbf{x}}^t}} \right)} \right\|^2},
\end{align}
where
\begin{align}
    \label{def k3}
    {\kappa _3} \triangleq \left[ {\frac{{4\delta }}{{Hd}} + 4\frac{{\left( {N - H} \right)\left( {N - d} \right)}}{{dH\left( {N - 1} \right)N}}} \right]N{\beta ^2},
\end{align}
\begin{align}
    \label{def k4}
    {\kappa _4} \triangleq 2\frac{1}{{{N^2}}} + \frac{{4\delta }}{{HdN}} + 4\frac{{\left( {N - H} \right)\left( {N - d} \right)}}{{dH\left( {N - 1} \right){N^2}}}.
\end{align}
\end{lemma}
\begin{proof}
    Please see Appendix~\ref{appen lemma4} in the supplementary materials. 
\end{proof}
Next, based on the above lemmas, we first present the main theorem that characterizes the convergence performance of Com-LAD. Based on this, the convergence performance of LAD can be derived accordingly as a special case of Com-LAD. 
\begin{theorem}[Convergence performance of Com-LAD]
\label{convergence performance}
If we fix the learning rate as ${\gamma ^t} = {\gamma ^0}  < \frac{{\frac{1}{N} - \sqrt {\kappa {\kappa _2}} }}{{L\kappa {\kappa _2} + L{\kappa _4}}}$, under the condition $\sqrt {\kappa {\kappa _2}}  < \frac{1}{N}$, Com-LAD converges as follows: 
\begin{align}
    \label{converg cola}
& \frac{1}{T}\sum\limits_{t = 0}^{T - 1} {\mathbb{E}\left[ {{{\left\| {\nabla F\left( {{{\mathbf{x}}^t}} \right)} \right\|}^2}} \right]}\nonumber\\
\leqslant & \frac{1}{T}\frac{{F\left( {{{\mathbf{x}}^0}} \right) - {F^*}}}{{{\gamma ^0}\left( {\frac{1}{N} - \sqrt {\kappa {\kappa _2}} } \right) - {{\left( {{\gamma ^0}} \right)}^2}\left( {L\kappa {\kappa _2} + L{\kappa _4}} \right)}} \nonumber\\
&+ \frac{{\frac{{{\kappa _1}\sqrt \kappa  }}{{2\sqrt {{\kappa _2}} }} + {\gamma ^0}\left( {L\kappa {\kappa _1} + L{\kappa _3}} \right)}}{{\left( {\frac{1}{N} - \sqrt {\kappa {\kappa _2}} } \right) - {\gamma ^0}\left( {L\kappa {\kappa _2} + L{\kappa _4}} \right)}}.
\end{align}
\end{theorem}
\begin{proof}
    Please see Appendix~\ref{appen proof the} in the supplementary materials. 
\end{proof}

\begin{theorem}[Convergence performance of LAD]
\label{convergence performance lad}
If we fix the learning rate as ${\gamma ^t} = {\gamma ^0}  < \frac{{\frac{1}{N} - \sqrt {\kappa {\xi _2}} }}{{L\kappa {\xi _2} + L{\xi _4}}}$, under the condition $\sqrt {\kappa {\xi _2}}  < \frac{1}{N}$, LAD converges as follows: 
\begin{align}
    \label{converg cola lad}
&\frac{1}{T}\sum\limits_{t = 0}^{T - 1} {\mathbb{E}\left[ {{{\left\| {\nabla F\left( {{{\mathbf{x}}^t}} \right)} \right\|}^2}} \right]} \nonumber\\ \leqslant& \frac{1}{T}\frac{{F\left( {{{\mathbf{x}}^0}} \right) - {F^*}}}{{{\gamma ^0}\left( {\frac{1}{N} - \sqrt {\kappa {\xi _2}} } \right) - {{\left( {{\gamma ^0}} \right)}^2}\left( {L\kappa {\xi _2} + L{\xi _4}} \right)}}\nonumber\\
&+ \frac{{\frac{{{\xi _1}\sqrt \kappa  }}{{2\sqrt {{\xi _2}} }} + {\gamma ^0}\left( {L\kappa {\xi _1} + L{\xi _3}} \right)}}{{\left( {\frac{1}{N} - \sqrt {\kappa {\xi _2}} } \right) - {\gamma ^0}\left( {L\kappa {\xi _2} + L{\xi _4}} \right)}},
\end{align}
where
\begin{align}
    \label{def xi1}
    {\xi _1} \triangleq 4{\beta ^2}\frac{{\left( {N - d} \right)N}}{{dH\left( {N - 1} \right)}},
\end{align}
\begin{align}
    \label{def xi2}
   {\xi _2} \triangleq \left[ {4\frac{{\left( {N - H} \right)\left( {N - d} \right)}}{{dH\left( {N - 1} \right)N}}} \right]\frac{1}{N},
\end{align}
\begin{align}
    \label{def xi3}
    {\xi _3} \triangleq 8\frac{{\left( {N - H} \right)\left( {N - d} \right)}}{{dH\left( {N - 1} \right)}}{\beta ^2},
\end{align}
\begin{align}
    \label{def xi4}
    {\xi _4} \triangleq 2\frac{1}{{{N^2}}}  + 8\frac{{\left( {N - H} \right)\left( {N - d} \right)}}{{dH\left( {N - 1} \right){N^2}}}.
\end{align}

\end{theorem}
\begin{proof}
    Note that when $\delta=0$, there is no communication compression. Based on this fact, Theorem~\ref{convergence performance lad} can be derived by substituting $\delta=0$ into Theorem~\ref{convergence performance}. 
\end{proof}
In (\ref{converg cola}) in Theorem~\ref{convergence performance}, the first term on the right-hand side diminishes as the number of iterations increases, while the second term represents the error term, given as
\begin{align}
    \label{error term}
   \varepsilon_\text{Com-LAD}  = \frac{{\frac{{{\kappa _1}\sqrt \kappa  }}{{2\sqrt {{\kappa _2}} }} + {\gamma ^0}\left( {L\kappa {\kappa _1} + L{\kappa _3}} \right)}}{{\left( {\frac{1}{N} - \sqrt {\kappa {\kappa _2}} } \right) - {\gamma ^0}\left( {L\kappa {\kappa _2} + L{\kappa _4}} \right)}},
\end{align}
where decreasing the value of $\kappa$ indicates better learning performance. Based on Definition~\ref{def:resaveraging}, a decreasing value of \(\kappa\) can be achieved by utilizing more advanced robust aggregation rules to enhance Byzantine robustness. 
By setting $d = O\left( N \right)$, when the number of devices is very large and a certain fraction are honest, the error term can be rewritten as
\begin{align}
    \label{error term com re}
   \varepsilon_\text{Com-LAD} = O\left( {\frac{{{\kappa _1}\sqrt \kappa  }}{{\sqrt {{\kappa _2}} }}}  \right).
\end{align}
\begin{figure}
    \centering
    \includegraphics[width=\linewidth]{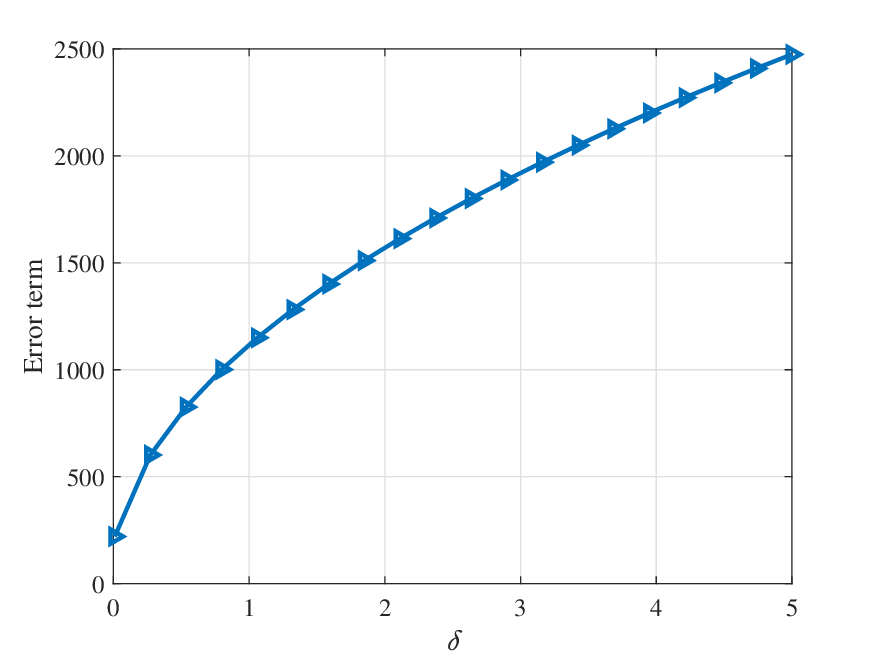}
    \caption{The error term as a function of $\delta$.}
    \label{fig: com_error_term delta}
\end{figure}
\begin{figure}
    \centering
    \includegraphics[width=\linewidth]{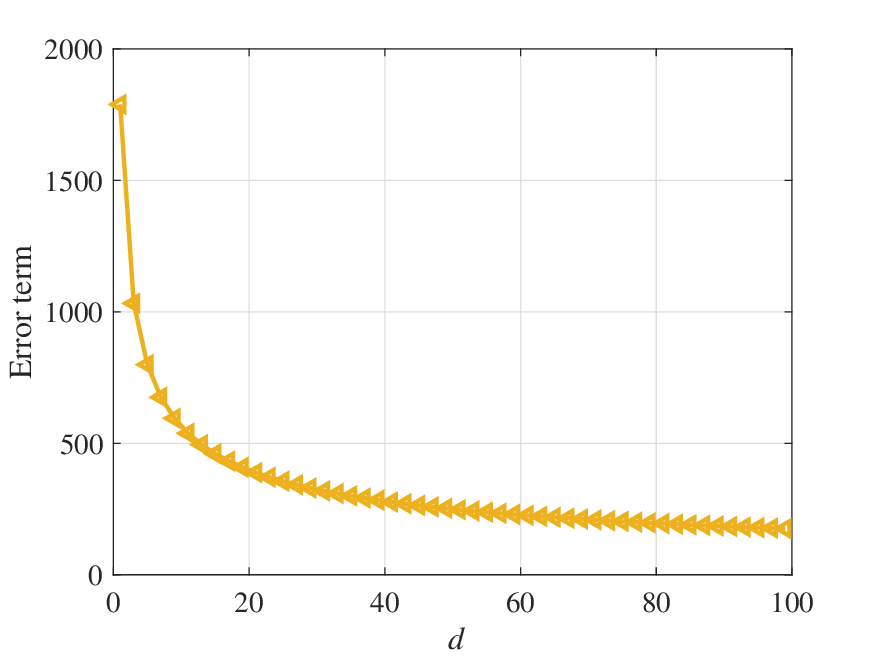}
    \caption{The error term as a function of $d$.}
    \label{fig: com_error_term d}
\end{figure}
As an illustrative example, when \( N = 100, H = 65, \kappa = 1.5, \beta = 1, d = 5 \), we plot the error term in~(\ref{error term com re}) as a function of \( \delta \) in Fig.~\ref{fig: com_error_term delta}. It can be observed that as \( \delta \) decreases, the error is reduced, while the communication overhead increases, highlighting a trade-off between learning performance and communication burden. 
In Fig.~\ref{fig: com_error_term d}, we illustrate the error term in~(\ref{error term com re}) as a function of \( d \), where \( N = 100, H = 65, \kappa = 1.5, \beta = 1, \delta = 0.5 \). The results show that increasing \( d \) leads to a reduction in the error term, which aligns with our intuition. As \( d \) increases, each device utilizes more subsets to compute local gradients in a single iteration, thereby enhancing learning performance at the cost of higher computational load on the devices.

In (\ref{converg cola lad}) in Theorem~\ref{convergence performance lad}, the error term can be expressed as
\begin{align}
    \label{error term lad}
  {\varepsilon _{{\text{LAD}}}} = \frac{{N\beta \frac{{\sqrt {\kappa {\xi _1}} }}{2}\sqrt {\frac{N}{{N - H}}}  + {\gamma ^0}\left( {L\kappa {\xi _1} + L{\xi _3}} \right)}}{{\left( {\frac{1}{N} - \sqrt {\kappa {\xi _2}} } \right) - {\gamma ^0}\left( {L\kappa {\xi _2} + L{\xi _4}} \right)}}.
\end{align}
In this scenario, by setting $d = O\left( N \right)$, when the number of devices is very large and a certain fraction are honest, the error term can be rewritten as
\begin{align}
    \label{error term approach}
   \varepsilon_\text{LAD} = O\left( {\beta ^2}\sqrt {\kappa \frac{{\left( {N - d} \right)N}}{{dH\left( {N - H} \right)}}}   \right).
\end{align}
Note that the error term derived in \cite{allouah2023fixing} can be expressed as:
\begin{align}
    \label{error term baseline}
  {\varepsilon _0} = O\left( {{\beta ^2}\kappa} \right).
\end{align}
From the comparison between (\ref{error term approach}) and (\ref{error term baseline}), we observe that the LAD error can be reduced relative to the baseline when $d \geqslant \frac{{{N^2}}}{{\kappa H\left( {N - H} \right) + N}}$. For example, when $N=100,H=65,\kappa=1.5$, a lower error can be attained by the proposed method by setting $d \geqslant 3$, which does not impose significant additional computational burden on the devices compared to the baseline. Moreover, the error decreases further as \(d\) increases. In the limiting case \(d = N\), the error term vanishes entirely, and LAD converges to the local optimum even in the presence of Byzantine attacks. 



\section{Numerical Results}
\label{simulations}
In this section, we evaluate the performance of the proposed methods on a linear regression task, under both settings with and without compressed communication. Suppose there are a total number of \(N=100\) devices. The training loss function is given as
\begin{align}
\label{LR loss}
F\left( {\mathbf{x}} \right)= \sum\limits_{k = 1}^N {{f_k}\left( {\mathbf{x}} \right)},
{f_k}\left( {\mathbf{x}} \right) = \frac{1}{2}{\left( {\left\langle {{\mathbf{x}},{{\mathbf{z}}_k}} \right\rangle  - {y_k}} \right)^2},
\end{align}
where \({{\mathbf{z}}_k} \in {\mathbb{R}^{100}}\), \(y_k \in {\mathbb{R}}\), and \(\mathbf{x} \in {\mathbb{R}^{100}}\). Here, the training dataset \(\mathcal{D}\) contains \(N=100\) data samples and is divided into \(N\) subsets, where each subset contains a single training data sample. In (\ref{LR loss}), \(\left\{ {{{\mathbf{z}}_1}, \dots, {{\mathbf{z}}_m}} \right\}\) are generated as random vectors where the elements therein are drawn independently from the Gaussian distribution \(\mathcal{N}\left( {0,100} \right)\). To generate the label \(y_k\) under data heterogeneity among the subsets, random vectors \(\mathbf{\overset{\lower0.5em\hbox{$\smash{\scriptscriptstyle\frown}$}}{x} }_k\) are generated for each subset $k$ in the following way. In \(\mathbf{\overset{\lower0.5em\hbox{$\smash{\scriptscriptstyle\frown}$}}{x} }_k\), the elements are drawn from the Gaussian distribution \(\mathcal{N}\left( {0,1+k\sigma_H} \right)\) and the label \(y_k\) can be generated as \(y_k \sim \mathcal{N}\left( {\left\langle {{{\mathbf{z}}_k},{\mathbf{\overset{\lower0.5em\hbox{$\smash{\scriptscriptstyle\frown}$}}{x} }}_k} \right\rangle ,1} \right)\), \(\forall k\). From the above process, it can be seen that a larger value of $\sigma_H$ indicates that the training data among the subsets are more heterogeneous. As a special case, when $\sigma_H=0$, the training data in each subsets are independent and identically distributed (IID). 

\subsection{Results without Compressed Communication}
Under the setting without compressed communication, the Byzantine devices execute the Sign-flipping attack \cite{zhu2023byzantine}, where the messages sent by the Byzantine devices are multiplied by a negative coefficient before the transmission and the coefficient is fixed as \(-2\). The number of the honest devices is set as $H=80$, unless specified, and the rest of the devices are Byzantine devices.  We compare our approach against the following baseline methods, which include the state-of-the-art methods: 

\begin{itemize} \item \textbf{Vanilla averaging (VA):} The training data are distributed to the devices in a non-redundant manner, and averaging-based aggregation is employed at the server to aggregate the local gradients sent by the honest devices and the messages from the Byzantine devices. 
\item \textbf{Coordinate-wise trimmed mean (CWTM)\cite{yin2018byzantine}:} The training data are distributed to the devices in a non-redundant manner. The server applies CWTM as the aggregation rule to aggregate the local gradients sent by the honest devices and the messages transmitted by the Byzantine devices. 
\item \textbf{CWTM with nearest neighbor mixing (CWTM-NNM)\cite{allouah2023fixing}:} The training data are distributed to the devices in a non-redundant manner. The server applies CWTM with nearest neighbor mixing (NNM) as the pre-aggregation step to aggregate the local gradients sent by the honest devices and the messages transmitted by the Byzantine devices. 
\item \textbf{Robust DT based on coding theory (DRACO)\cite{chen2018draco}:} Each device computes redundant local gradients and sends coded vectors to the server to fully eliminate the effects of Byzantine attacks with problem-independent robustness. This train can attain the same model as if in the adversary-free setup. \end{itemize}

Here, we adopt CWTM as the baseline method instead of other robust aggregation rules. This choice is motivated by two key factors shown in \cite{allouah2023fixing}: (i) Theoretically, CWTM has been shown to exhibit near-optimal Byzantine resilience when the number of Byzantine devices is small, and (ii) it consistently outperforms other robust aggregation rules in experiments, such as the geometric median and coordinate-wise median, in terms of learning performance under Byzantine attacks. In addition, it has been shown in \cite{allouah2023fixing} that when combined with NNM, any standard robust aggregation rule can achieve optimal Byzantine resilience among all the methods that aim to design robust aggregation rules to replace the averaging-based
aggregation rule at the server. Building on this result, we adopt CWTM-NNM as another baseline method, where NNM is applied as a pre-aggregation step at the server prior to executing CWTM. Among the gradient coding methods designed to address Byzantine attacks in DT, we select DRACO as the baseline method. This choice is motivated by two main reasons: (i) DRACO incurs no additional communication overhead during training compared to the VA method, and (ii) it is applicable to a broader range of DT settings beyond matrix multiplication tasks. For a fair comparison, we evaluate the performance of the baseline methods alongside the proposed LAD method, using CWTM and CWTM with NNM as the aggregation rules at the server. Under these two configurations, the proposed method is referred to as LAD-CWTM and LAD-CWTM-NNM, respectively. This naming reflects the fact that LAD serves as a meta-algorithm that can be combined with any robust aggregation rule.

To compare the learning performance of different methods under Byzantine attacks, Fig.~\ref{fig: compare_linear} presents the training loss as a function of the number of iterations for each method. In this setting, the learning rate is set to $\gamma^t = 0.000001$, the parameter for CWTM is set to $0.1$, and $\sigma_H = 0.3$. 
For a fair comparison, it is assumed that the entire set of training subsets is distributed to all devices before training begins, for both the baseline methods and the proposed methods. In the baseline methods, each device computes the local gradient corresponding to a single randomly selected subset, without computational redundancy among devices. This setup ensures that the learning performance is not affected by variability or persistence in the identities of the devices. It is equivalent to allocating computational tasks in LAD by setting $d = 1$. For the proposed method, we evaluate its performance under different values of $d$, where a larger value of $d$ indicates a higher computational burden on each device. 

From Fig.~\ref{fig: compare_linear}, it can be observed that VA exhibits significantly deteriorated learning performance. This is because it employs an averaging-based aggregation rule at the server, which is highly susceptible to the influence of messages sent by Byzantine devices. As a result, the applied model updates can deviate substantially from the true global gradient. 
By comparing CWTM and LAD-CWTM, we observe that the latter achieves better learning performance. This improvement is attributed to the fact that LAD-CWTM leverages computational redundancy among devices to enhance resilience against Byzantine attacks. Moreover, as the value of $d$ increases, the performance gain of LAD-CWTM becomes more pronounced, albeit at the cost of increased computational burden on the devices.
Likewise, LAD-CWTM-NNM outperforms CWTM-NNM for similar reasons, demonstrating the advantages of the proposed method under various robust aggregation rules at the server.
Comparing LAD-CWTM and LAD-CWTM-NNM, we further observe that LAD-CWTM-NNM yields enhanced learning performance. This improvement results from incorporating NNM as a pre-aggregation step at the server, which strengthens the robustness of the aggregation rule. Consequently, the overall learning performance of the proposed method is improved, as supported by Theorem~\ref{convergence performance lad}. 
In Fig.~\ref{fig: compare_linear}, we also present the performance of DRACO for reference. It can be seen that DRACO achieves the best learning performance, which aligns with our intuition that DRACO can fully eliminate the effects of Byzantine attacks, thereby attaining the same performance as in the absence of such attacks. 
However, it is important to note that this performance gain comes at the cost of a very high computational burden on each device. According to \cite{chen2018draco}, each device in DRACO is required to compute $41$ local gradients, which is equivalent to the computational burden in the proposed method when $d = 41$.
Based on these observations, the proposed method demonstrates an advantageous trade-off between computational burden and resilience to Byzantine attacks. Notably, when $d = 20$, the proposed method achieves learning performance that is very close to that of DRACO, while incurring only half of its computational burden.
\begin{figure}
    \centering
    \includegraphics[width=\linewidth]{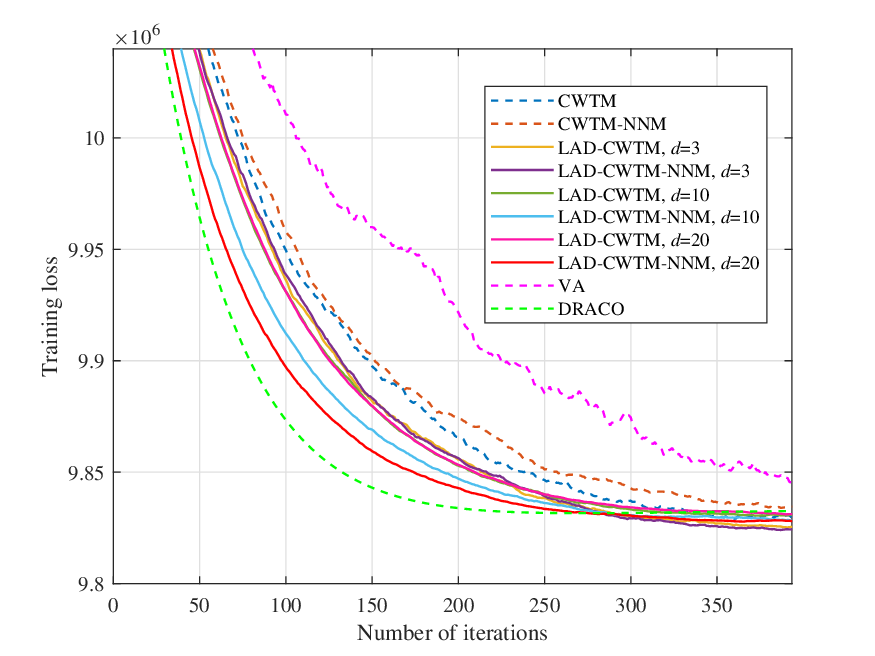}
    \caption{The training loss as a function of number of iterations for different methods.}
    \label{fig: compare_linear}
\end{figure}

To evaluate the performance of the proposed method under different levels of data heterogeneity among the subsets, Fig.~\ref{fig: heter} presents the training loss as a function of the number of iterations for various methods under different values of $\sigma_H$. In this setting, the number of Byzantine devices is fixed at $20$, the learning rate is set to $\gamma^t = 0.000001$, the parameter for CWTM is set to $0.1$, and we fix $d=10$ in the proposed method. 
We omit the performance curves of VA and DRACO due to the severe degradation of VA under Byzantine attacks and the substantial computational burden incurred by DRACO. From Fig.~\ref{fig: heter}, it can be observed that LAD consistently outperforms the baseline methods for the same value of $\sigma_H$ and when using the same robust aggregation rule at the server.
This performance advantage of LAD is more pronounced when $\sigma_H = 0.1$ compared to the case when $\sigma_H = 0$, indicating that the superiority of LAD becomes more evident as the heterogeneity among training subsets increases. This observation aligns with the motivation of this work: existing robust aggregation methods suffer from a non-diminishing solution error that worsens with increased data heterogeneity, whereas the proposed method mitigates this error through computational redundancy across devices.

\begin{figure}[htbp]
    \centering
    \includegraphics[width=\linewidth]{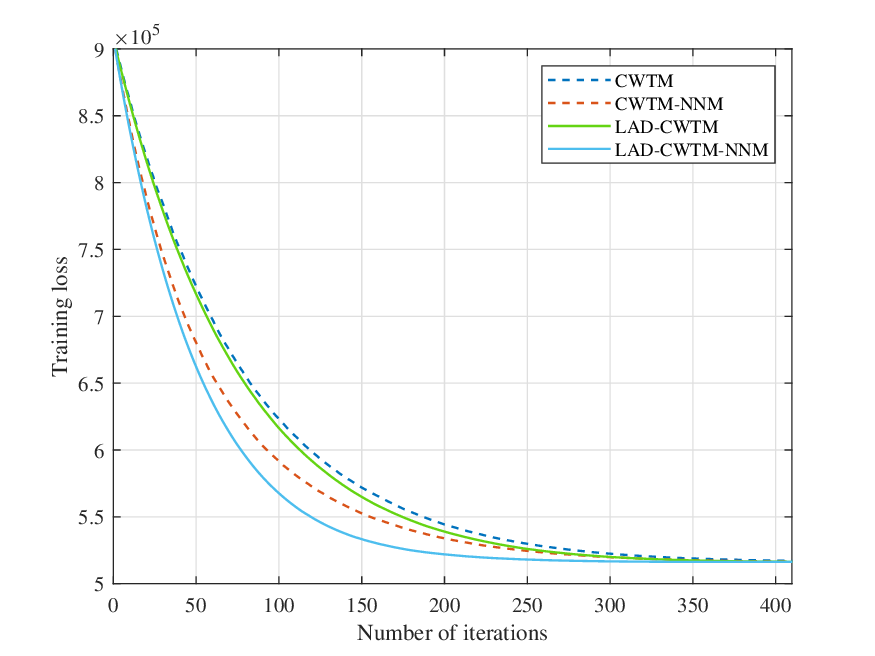}\\
    (a) \\[1em]  

    \includegraphics[width=\linewidth]{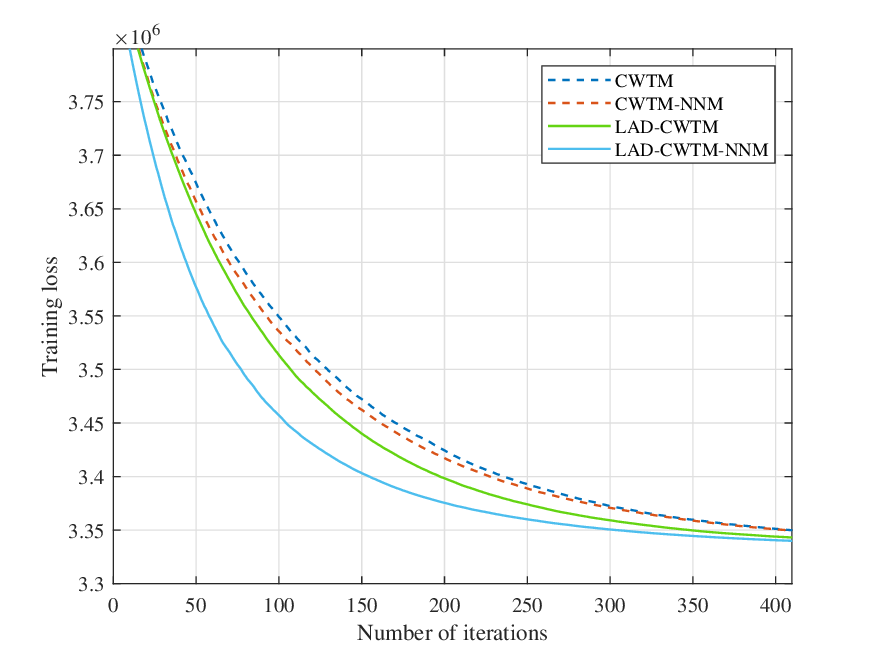}\\
    (b) 

    \caption{The training loss as a function of number of iterations for different methods under different values of $\sigma_H$. (a) $\sigma_H = 0$. (b) $\sigma_H = 0.1$.}
    \label{fig: heter}
\end{figure}

\subsection{Results with Compressed Communication}
Under the compressed communication setting, our method employs random sparsification as the compression function. In this case, each vector transmitted by the honest devices contains only \( \hat{Q} \) non-zero elements, with their indices randomly and uniformly selected.
 Based on this setup, the Byzantine devices perform a Sign-flipping attack: their messages are first multiplied by a fixed negative coefficient of $-2$ and then compressed using random sparsification before transmission. The number of honest devices is set to \( H = 70 \). For the proposed LAD method, we adopt CWTM and CWTM with NNM as the aggregation rules at the server, representing two distinct implementations. Under these configurations, our method is referred to as Com-LAD-CWTM and Com-LAD-CWTM-NNM, respectively.

We compare our approach against the following baseline methods, which include several state-of-the-art techniques:

\begin{itemize} 
\item \textbf{Vanilla averaging with compression (Com-VA):} The training data are distributed to the devices in a non-redundant manner, and averaging-based aggregation is employed at the server to aggregate the compressed version of the local gradients sent by the honest devices and the messages from the Byzantine devices. 
\item \textbf{Coordinate-wise trimmed mean with compression (Com-CWTM):} The training data are distributed to the devices in a non-redundant manner. The server applies CWTM as the aggregation rule to aggregate the compressed version of the local gradients sent by the honest devices and the messages transmitted by the Byzantine devices. This method originates from \cite{zhu2023byzantine}, where geometric median-based robust aggregation is replaced with CWTM. 
\item \textbf{CWTM with NNM with compression (Com-CWTM-NNM)}: The training data are distributed to the devices in a non-redundant manner. The server applies CWTM with NNM as the pre-aggregation step to aggregate the compressed version of the local gradients sent by the honest devices and the messages transmitted by the Byzantine devices. 
\item \textbf{Robust DT via compression and thresholding based on gradient norm (Com-TGN) \cite{ghosh2021communication}:} The training data are distributed to the devices in a non-redundant manner. The server applies thresholding based on gradient norm as the aggregation rule to aggregate the compressed version of the local gradients sent by the honest devices and the messages transmitted by the Byzantine devices. 
\end{itemize}

Note that among the existing methods addressing both Byzantine robustness and communication efficiency in DT, we do not include those that rely on variance reduction techniques, such as the methods proposed in \cite{rammal2024communication} and \cite{gorbunov2022variance}. This exclusion is due to the fact that variance reduction techniques require access to historical information from individual devices. In scenarios where device identities may vary over time, these techniques become inapplicable, as devices that are currently honest may have behaved as Byzantine in the past. Consequently, their historical information could be incorrect and misleading. In addition, unlike the setting without compressed communication, DRACO in \cite{chen2018draco} is not adopted as a baseline method here, as it is incompatible with communication compression.

To compare the learning performance of different methods with compressed communication under Byzantine attacks, Fig.~\ref{fig: compare_compress} presents the training loss as a function of the number of iterations for each method. In this experiment, the learning rate is set to $\gamma^t = 0.0000003$, the parameter for CWTM is $0.1$, the parameter for Com-TGN is $0.2$, the proposed method uses $d = 3$, $\sigma_H$ is fixed at $0.3$ and we set $\hat{Q}=30$.
From Fig.~\ref{fig: compare_compress}, we observe that Com-VA exhibits the worst learning performance. This is because the averaging-based aggregation at the server in the compressed domain based on all devices can be easily manipulated by the incorrect messages sent by the Byzantine devices. 
Comparing Com-LAD-CWTM and Com-LAD-CWTM-NNM, we observe that the latter achieves better performance. This is because incorporating NNM as a pre-aggregation step at the server enhances the robustness of the aggregation rule in the compressed domain, thereby improving the overall resilience of the system to Byzantine attacks. This observation aligns with our theoretical results in Section~\ref{performance analysis}, which suggest that employing more advanced robust aggregation rules in Com-LAD can enhance Byzantine robustness. 
Furthermore, comparing Com-CWTM and Com-LAD-CWTM, we observe that the latter attains better learning performance. This improvement is attributed to the ability of Com-LAD-CWTM to encode local gradients prior to compression, leveraging computational redundancy among devices to enhance robustness against Byzantine attacks. Likewise, the comparison between Com-CWTM-NNM and Com-LAD-CWTM-NNM shows a more pronounced performance gain for the latter, highlighting that utilizing computational redundancy can further enhance the resilience of the server-side aggregation rule.
In Fig.~\ref{fig: compare_compress}, it can be observed that the learning performance of Com-TGN is slightly better than that of Com-LAD-CWTM, but inferior to Com-LAD-CWTM-NNM. This is attributed to the fact that the aggregation rule in Com-TGN is more advanced when applied in the compressed domain compared to CWTM, which is not specifically designed for aggregation in the compressed domain. However, when combined with NNM, Com-LAD-CWTM-NNM clearly outperforms Com-TGN. This performance gain stems not only from the inherent advantages of NNM, but also from the enhanced robustness of the aggregation process, which leverages computational redundancy and encoding techniques on the devices.

It is worth noting that the improved performance of the proposed method comes at the cost of increased computational burden on the devices. However, this overhead is negligible, as even a small value of $d$—indicating a low level of computational redundancy—yields significant performance gains. This demonstrates that the proposed method achieves a favorable trade-off between computational efficiency and robustness against Byzantine attacks under compressed communication.
\begin{figure}
    \centering
    \includegraphics[width=\linewidth]{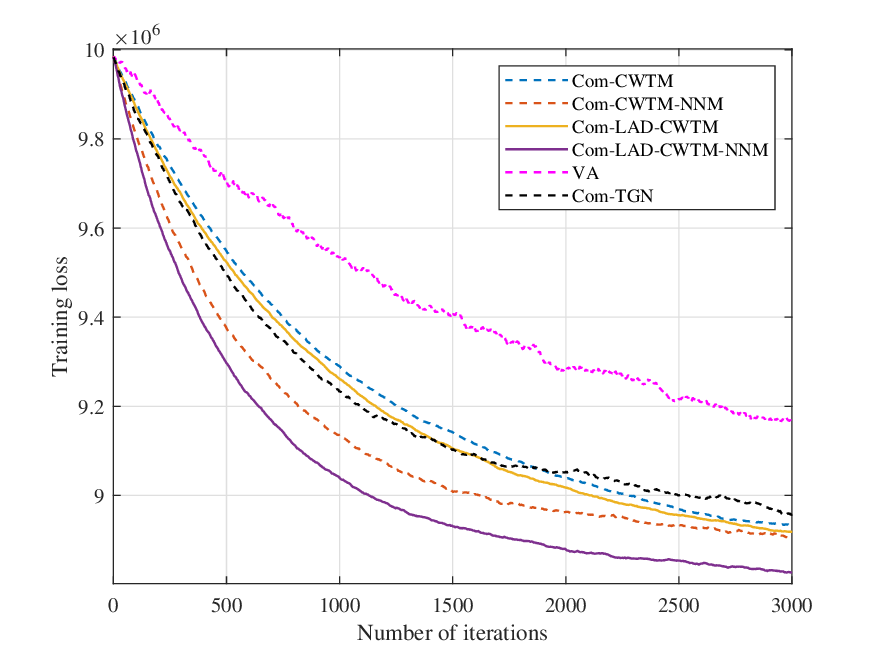}
    \caption{The training loss as a function of number of iterations for different methods with compressed communication.}
    \label{fig: compare_compress}
\end{figure}
\section{Conclusions}
\label{conclusions}
This paper investigated the problem of DT under Byzantine attacks in the presence of communication constraints. We identified a key limitation in existing robust aggregation schemes: the persistence of non-diminishing solution error when local gradients vary significantly across devices. To address this, we proposed a novel method LAD which leverages redundant gradient computation and encoding to enhance robustness. Theoretical analysis demonstrated that LAD achieves improved convergence guarantees and mitigates the impact of Byzantine devices more effectively.
To further improve communication efficiency, we extend LAD to Com-LAD by incorporating encoding and communication compression on the devices. Under this compressed communication setting, Com-LAD demonstrates improved robustness compared to the baseline methods, while significantly reducing the communication overhead in each iteration.
Extensive experimental results validated the effectiveness of both LAD and Com-LAD in improving resilience to Byzantine attacks and maintaining high learning performance under constrained communication. 
\bibliographystyle{ieeetr} 
\bibliography{reference}   

\newpage
\newpage
\appendices
\section{Proof of Lemma~\ref{lemma1}}
\label{appen lemma1}
Based on the fact that each row in $\mathbf{S}$ has exactly $d$ entries equal to one and all other entries are zero, we can easily derive that 
\begin{align}
    \label{re S}
   & \mathbb{E}\left( {{{\left\| {\left( {\frac{1}{d}\frac{1}{H}{{\mathbf{h}}^{1 \times N}}{\mathbf{S}} - \frac{1}{N}{{\mathbf{1}}^{1 \times N}}} \right)} \right\|}^2}} \right)\nonumber\\
     = & \mathbb{E}\left( {{{\left\| {\frac{1}{d}\frac{1}{H}{{\mathbf{h}}^{1 \times N}}{\mathbf{S}}} \right\|}^2}} \right) + \mathbb{E}\left( {{{\left\| {\frac{1}{N}{{\mathbf{1}}^{1 \times N}}} \right\|}^2}} \right)\nonumber\\
     &- 2\mathbb{E}\left( {\frac{1}{d}\frac{1}{H}\frac{1}{N}{{\mathbf{h}}^{1 \times N}}{\mathbf{S}}{{\mathbf{1}}^{N \times 1}}} \right)\nonumber\\
      = & \frac{1}{{{d^2}{{H}^2}}}\mathbb{E}\left( {{{\left\| {{{\mathbf{h}}^{1 \times N}}{{\mathbf{S}}}} \right\|}^2}} \right) + \frac{1}{N} - 2\frac{1}{N}  \nonumber\\
       = & \frac{1}{{{d^2}{{H}^2}}}\sum\limits_{j = 1}^N {\mathbb{E}\left[ {{{\left( {\sum\limits_{i = 1}^N {{h_i}s\left( {i,j} \right)} } \right)}^2}} \right]}  - \frac{1}{N},
\end{align}
where $h_i$ is the $i$-th element in $\mathbf{h}$.  
In (\ref{re S}), we have
\begin{align}
    \label{re S1}
    &\mathbb{E}\left[ {{{\left( {\sum\limits_{i = 1}^N {{h_i}s\left( {i,j} \right)} } \right)}^2}} \right]\nonumber\\
    =& \sum\limits_{i = 1}^N {\mathbb{E}\left[ {{{\left( {{h_i}s\left( {i,j} \right)} \right)}^2}} \right]}  + \sum\limits_{i \ne k} {\mathbb{E}\left[ {{h_i}s\left( {i,j} \right){h_k}s\left( {k,j} \right)} \right]} \nonumber\\
    = & \sum\limits_{i = 1}^N {s\left( {i,j} \right)\mathbb{E}\left( {{h_i}} \right)}  + \sum\limits_{i \ne k} {s\left( {i,j} \right)s\left( {k,j} \right)\mathbb{E}\left[ {{h_i}{h_k}} \right]} \nonumber\\
    =& \frac{H}{N}\sum\limits_{i = 1}^N {s\left( {i,j} \right)} + \frac{{H\left( {H - 1} \right)}}{{N\left( {N - 1} \right)}}\sum\limits_{i \ne k} {s\left( {i,j} \right)s\left( {k,j} \right)},
\end{align}
where the last equality holds due to the randomness of $\mathbf{h}$. 
Substituting (\ref{re S1}) into (\ref{re S}), we have
\begin{align}
    \label{re S2}
 & \mathbb{E}\left( {{{\left\| {\left( {\frac{1}{d}\frac{1}{H}{{\mathbf{h}}^{1 \times N}}{\mathbf{S}} - \frac{1}{N}{{\mathbf{1}}^{1 \times N}}} \right)} \right\|}^2}} \right) \nonumber\\
   = & \frac{1}{{{d^2}{H^2}}}\left[ {Hd + \frac{{H\left( {H - 1} \right)}}{{N\left( {N - 1} \right)}}\sum\limits_{j = 1}^N {\sum\limits_{i \ne k} {s\left( {i,j} \right)s\left( {k,j} \right)} } } \right] - \frac{1}{N}.
\end{align}
Suppose there are $\theta_j$ ones in the $j$-th column of matrix $\mathbf{S}$. Then we can rewrite $\sum\limits_{j = 1}^N {\sum\limits_{i \ne k} {s\left( {i,j} \right)s\left( {k,j} \right)} }$ in (\ref{re S2}) as 
\begin{align}
    \label{re S3}
    \sum\limits_{j = 1}^N {\sum\limits_{i \ne k} {s\left( {i,j} \right)s\left( {k,j} \right)} }  = 2\sum\limits_{j = 1}^N {\left( {\begin{array}{*{20}{c}}
  {{\theta _j}} \\ 
  2 
\end{array}} \right)}  = \sum\limits_{j = 1}^N {\theta _j^2}  - dN.
\end{align}
Note that 
\begin{align}
    \label{re S4}
    {\inf _{\left\{ {\left. {{\theta _j},\forall j} \right|\sum\limits_{j = 1}^N {{\theta _j}}  = dN} \right\}}}\sum\limits_{j = 1}^N {\theta _j^2}  - dN = {d^2}N - dN,
\end{align}
where this infimum can be attained when $\theta_1=\dots=\theta_N=d$. According to (\ref{re S2})-(\ref{re S4}), we have
\begin{align}
    \label{re S5}
   & {\inf _{{\mathbf{S}} \in \mathcal{S}}}\mathbb{E}\left( {{{\left\| {\left( {\frac{1}{d}\frac{1}{H}{{\mathbf{h}}^{1 \times N}}{\mathbf{S}} - \frac{1}{N}{{\mathbf{1}}^{1 \times N}}} \right)} \right\|}^2}} \right)\nonumber\\
    = & \frac{{\left( {N - H} \right)\left( {N - d} \right)}}{{dH\left( {N - 1} \right)N}}.
\end{align}
By constructing the matrix $\mathbf{S}$ as $\mathbf{\hat{S}}$, it holds that $\theta_1=\dots=\theta_N=d$. Based on that, the infimum in (\ref{re S5}) can be attained by setting $\mathbf{S}$ as $\mathbf{\hat{S}}$, which completes the proof. 
\section{Proof of Lemma~\ref{lemma20}}
\label{appen lemma20}
Based on (\ref{encoding}) and the randomness of the task indices received by the devices, it holds that 
\begin{align}
    \label{mean git}
    \mathbb{E}\left( {\left. {{\mathbf{g}}_i^t} \right|{\mathcal{F}^t}} \right) = {{\boldsymbol{\mu }}^t}.
\end{align}
According to (\ref{mean git}), we can derive that 
\begin{align}
    \label{lemma20-1}
    \mathbb{E}\left( {\left. {{{\left\| {{\mathbf{g}}_i^t - {{\boldsymbol{\mu }}^t}} \right\|}^2}} \right|{\mathcal{F}^t}} \right) = \mathbb{E}\left( {\left. {{{\left\| {{\mathbf{g}}_i^t} \right\|}^2}} \right|{\mathcal{F}^t}} \right) - {\left\| {{{\boldsymbol{\mu }}^t}} \right\|^2}.
\end{align}
In (\ref{lemma20-1}), we have
\begin{align}
    \label{lemma20-2}
 & \mathbb{E}\left( {\left. {{{\left\| {{\mathbf{g}}_i^t} \right\|}^2}} \right|{\mathcal{F}^t}} \right) = \mathbb{E}\left( {\left. {{{\left\| {\frac{1}{d}\sum\limits_{k = 1}^d {{{\mathbf{y}}_k}} } \right\|}^2}} \right|{\mathcal{F}^t}} \right) \nonumber\\
  = & \frac{1}{{{d^2}}}\mathbb{E}\left( {\left. {\sum\limits_{k = 1}^d {{{\left\| {{{\mathbf{y}}_k}} \right\|}^2}}  + 2\sum\limits_{k < j}^{} {\left\langle {{{\mathbf{y}}_k},{{\mathbf{y}}_j}} \right\rangle } } \right|{\mathcal{F}^t}} \right) \nonumber \\
   = & \frac{1}{{{d^2}}}\left( {\sum\limits_{k = 1}^d {\mathbb{E}\left( {\left. {{{\left\| {{{\mathbf{y}}_k}} \right\|}^2}} \right|{\mathcal{F}^t}} \right)}  + 2\sum\limits_{k < j}^{} {\mathbb{E}\left[ {\left. {\left\langle {{{\mathbf{y}}_k},{{\mathbf{y}}_j}} \right\rangle } \right|{\mathcal{F}^t}} \right]} } \right),
\end{align}
where \(\left\{ {{{\mathbf{y}}_i},i = 1,...,d} \right\}\) are the local gradients computed by device $i$. 

In (\ref{lemma20-2}), we know 
\begin{align}
    \label{lemma20-3}
   & \mathbb{E}\left( {\left. {{{\left\| {{{\mathbf{y}}_k}} \right\|}^2}} \right|{\mathcal{F}^t}} \right)\nonumber \\
    = & \frac{1}{N}\sum\limits_{k = 1}^N {{{\left\| {\nabla {f_k}\left( {{{\mathbf{x}}^t}} \right)} \right\|}^2}} \nonumber \\
    = & \frac{1}{N}\sum\limits_{k = 1}^N {{{\left\| {\nabla {f_k}\left( {{{\mathbf{x}}^t}} \right) - {{\boldsymbol{\mu }}^t} + {{\boldsymbol{\mu }}^t}} \right\|}^2}} \nonumber \\
    =& \frac{1}{N}\sum\limits_{k = 1}^N {{{\left\| {\nabla {f_k}\left( {{{\mathbf{x}}^t}} \right) - {{\boldsymbol{\mu }}^t}} \right\|}^2}}  + {\left\| {{{\boldsymbol{\mu }}^t}} \right\|^2}. 
\end{align}
In addition, based on the randomness of the task indices received by the devices and the randomness of the vector $\mathbf{p}^t$, we have
\begin{align}
    \label{lemma20-4}
   & \sum\limits_{k < j}^{} {\mathbb{E}\left[ {\left. {\left\langle {{{\mathbf{y}}_k},{{\mathbf{y}}_j}} \right\rangle } \right|{\mathcal{F}^t}} \right]} \nonumber \\
   = & \left( {\begin{array}{*{20}{c}}
  d \\ 
  2 
\end{array}} \right)\frac{1}{{\left( {\begin{array}{*{20}{c}}
  N \\ 
  2 
\end{array}} \right)}}\sum\limits_{k < j}^{} {\mathbb{E}\left[ {\left. {\left\langle {\nabla {f_k}\left( {{{\mathbf{x}}^t}} \right),\nabla {f_j}\left( {{{\mathbf{x}}^t}} \right)} \right\rangle } \right|{\mathcal{F}^t}} \right]} \nonumber\\
=  & \frac{{d\left( {d - 1} \right)}}{{N\left( {N - 1} \right)}}\frac{{{{\left\| {\sum\limits_{k = 1}^N {\nabla {f_k}\left( {{{\mathbf{x}}^t}} \right)} } \right\|}^2} - \sum\limits_{k = 1}^N {{{\left\| {\nabla {f_k}\left( {{{\mathbf{x}}^t}} \right)} \right\|}^2}} }}{2}.
\end{align} 
By substituting (\ref{lemma20-3}) and (\ref{lemma20-4}) into (\ref{lemma20-2}), we have
\begin{align}
    \label{lemma20-5}
  &  \mathbb{E}\left( {\left. {{{\left\| {{\mathbf{g}}_i^t} \right\|}^2}} \right|{\mathcal{F}^t}} \right) = \frac{{\left( {d - 1} \right)}}{{dN\left( {N - 1} \right)}}{\left\| {\sum\limits_{k = 1}^N {\nabla {f_k}\left( {{{\mathbf{x}}^t}} \right)} } \right\|^2} \nonumber \\
  &+ \left[ {\frac{1}{{Nd}} - \frac{{\left( {d - 1} \right)}}{{dN\left( {N - 1} \right)}}} \right]\sum\limits_{k = 1}^N {{{\left\| {\nabla {f_k}\left( {{{\mathbf{x}}^t}} \right)} \right\|}^2}} \nonumber\\
     = & {\left\| {{{\boldsymbol{\mu }}^t}} \right\|^2} + \frac{1}{{Nd}}\frac{{N - d}}{{N - 1}}\left( {\sum\limits_{k = 1}^N {{{\left\| {\nabla {f_k}\left( {{{\mathbf{x}}^t}} \right) - {{\boldsymbol{\mu }}^t}} \right\|}^2}} } \right)\nonumber \\
     \leqslant & {\left\| {{{\boldsymbol{\mu }}^t}} \right\|^2} + \frac{1}{d}\frac{{N - d}}{{N - 1}}{\beta ^2}.
\end{align}
Substituting (\ref{lemma20-5}) back into (\ref{lemma20-1}) yields (\ref{h1}) in Lemma~\ref{lemma20}. 
\section{Proof of Lemma~\ref{lemma2}}
\label{appen lemma2}
By applying the basic inequality as
\begin{align}
    \label{basic ineq}
    \left\| \sum_{i=1}^{n} \mathbf{x}_i \right\|^2 \leqslant n \sum_{i=1}^{n} \left\| \mathbf{x}_i \right\|^2, &\forall\mathbf{x}_i \in \mathbb{R}^{Q\times1},
\end{align}
we can obtain 
    \begin{align}
        \label{bar hat}
 & {\left\| {{{{\mathbf{\bar g}}}^t} - {\mathbf{\hat g}}_i^t} \right\|^2} \leqslant 2{\left\| {{{{\mathbf{\bar g}}}^t} - {{\boldsymbol{\mu }}^t}} \right\|^2} + 2{\left\| {{\mathbf{\hat g}}_i^t - {{\boldsymbol{\mu }}^t}} \right\|^2}   \nonumber \\
  =& 2{\left\| {{{{\mathbf{\bar g}}}^t} - {{\boldsymbol{\mu }}^t}} \right\|^2} + 2{\left\| {{\mathbf{\hat g}}_i^t - {\mathbf{g}}_i^t + {\mathbf{g}}_i^t - {{\boldsymbol{\mu }}^t}} \right\|^2} \nonumber \\
   \leqslant& 2{\left\| {{{{\mathbf{\bar g}}}^t} - {{\boldsymbol{\mu }}^t}} \right\|^2} + 4{\left\| {{\mathbf{\hat g}}_i^t - {\mathbf{g}}_i^t} \right\|^2} + 4{\left\| {{\mathbf{g}}_i^t - {{\boldsymbol{\mu }}^t}} \right\|^2}.
    \end{align}
Based on (\ref{bar hat}), we have
\begin{align}
    \label{max bar hat}
    \frac{1}{H}\sum\limits_{i \in {\mathcal{H}^t}} {{{\left\| {{{{\mathbf{\bar g}}}^t} - {\mathbf{\hat g}}_i^t} \right\|}^2}}  \leqslant &2{\left\| {{{{\mathbf{\bar g}}}^t} - {{\boldsymbol{\mu }}^t}} \right\|^2} + 4\frac{1}{H}\sum\limits_{i \in {\mathcal{H}^t}} {{{\left\| {{\mathbf{\hat g}}_i^t - {\mathbf{g}}_i^t} \right\|}^2}} \nonumber \\
    &+ 4\frac{1}{H}\sum\limits_{i \in {\mathcal{H}^t}} {{{\left\| {{\mathbf{g}}_i^t - {{\boldsymbol{\mu }}^t}} \right\|}^2}}.
\end{align}
Taking conditional expectations on both sides of (\ref{max bar hat}), we can derive 
\begin{align}
    \label{exp bar hat}
    &\mathbb{E}\left( {\left. {\frac{1}{H}\sum\limits_{i \in {\mathcal{H}^t}} {{{\left\| {{{{\mathbf{\bar g}}}^t} - {\mathbf{\hat g}}_i^t} \right\|}^2}} } \right|{\mathcal{F}^t}} \right) \nonumber\\
    \leqslant& 2\mathbb{E}\left( {\left. {{{\left\| {{{{\mathbf{\bar g}}}^t} - {{\boldsymbol{\mu }}^t}} \right\|}^2}} \right|{\mathcal{F}^t}} \right) + 4\mathbb{E}\left( {\left. {\frac{1}{H}\sum\limits_{i \in {\mathcal{H}^t}} {{{\left\| {{\mathbf{\hat g}}_i^t - {\mathbf{g}}_i^t} \right\|}^2}} } \right|{\mathcal{F}^t}} \right)\nonumber \\
    &+ 4\mathbb{E}\left( {\left. {\frac{1}{H}\sum\limits_{i \in {\mathcal{H}^t}} {{{\left\| {{\mathbf{g}}_i^t - {{\boldsymbol{\mu }}^t}} \right\|}^2}} } \right|{\mathcal{F}^t}} \right)\nonumber\\
    =& 2\mathbb{E}\left( {\left. {{{\left\| {{{{\mathbf{\bar g}}}^t} - {{\boldsymbol{\mu }}^t}} \right\|}^2}} \right|{\mathcal{F}^t}} \right) + 4\mathbb{E}\left( {\left. {\frac{1}{H}\sum\limits_{i \in {\mathcal{H}^t}} {{{\left\| {{\mathbf{\hat g}}_i^t - {\mathbf{g}}_i^t} \right\|}^2}} } \right|{\mathcal{F}^t}} \right)\nonumber \\
    &+ 4\mathbb{E}\left( {\left. {{{\left\| {{\mathbf{g}}_i^t - {{\boldsymbol{\mu }}^t}} \right\|}^2}} \right|{\mathcal{F}^t}} \right), 
\end{align}
For the first term in (\ref{exp bar hat}), we have
\begin{align}
    \label{1st term}
   & \mathbb{E}\left( {\left. {{{\left\| {{{{\mathbf{\bar g}}}^t} - {{\boldsymbol{\mu }}^t}} \right\|}^2}} \right|{\mathcal{F}^t}} \right)\nonumber \\
    =& \mathbb{E}\left( {\left. {{{\left\| {\frac{1}{H}\sum\limits_{i \in {\mathcal{H}^t}} {{\mathbf{\hat g}}_i^t}  - \frac{1}{H}\sum\limits_{i \in {\mathcal{H}^t}} {{\mathbf{g}}_i^t}  + \frac{1}{H}\sum\limits_{i \in {\mathcal{H}^t}} {{\mathbf{g}}_i^t}  - {{\boldsymbol{\mu }}^t}} \right\|}^2}} \right|{\mathcal{F}^t}} \right)\nonumber\\
    \mathop  \leqslant \limits^{\left\langle 1 \right\rangle }  & 2\frac{1}{{{H^2}}}\mathbb{E}\left( {\left. {{{\left\| {\sum\limits_{i \in {\mathcal{H}^t}} {\left( {{\mathbf{\hat g}}_i^t - {\mathbf{g}}_i^t} \right)} } \right\|}^2}} \right|{\mathcal{F}^t}} \right)\nonumber \\
    &+ 2\mathbb{E}\left( {\left. {{{\left\| {\frac{1}{H}\sum\limits_{i \in {\mathcal{H}^t}} {{\mathbf{g}}_i^t}  - {{\boldsymbol{\mu }}^t}} \right\|}^2}} \right|{\mathcal{F}^t}} \right)\nonumber\\
     =& 2\frac{1}{{{H^2}}}\mathbb{E}\left[ {\left. {\sum\limits_{i \in {\mathcal{H}^t}} {\mathbb{E}\left( {\left. {{{\left\| {{\mathbf{\hat g}}_i^t - {\mathbf{g}}_i^t} \right\|}^2}} \right|{\mathcal{F}^t},\left\{ {\mathcal{T}_i^t},{{\mathbf{p}}^t} \right\}} \right)} } \right|{\mathcal{F}^t}} \right] \nonumber \\
     &+ 2\mathbb{E}\left( {\left. {{{\left\| {\frac{1}{H}\sum\limits_{i \in {\mathcal{H}^t}} {{\mathbf{g}}_i^t}  - {{\boldsymbol{\mu }}^t}} \right\|}^2}} \right|{\mathcal{F}^t}} \right)\nonumber\\
      \mathop  \leqslant \limits^{\left\langle 2 \right\rangle } & 2\frac{1}{{{H^2}}}\delta {\mathbb{E}_{\left\{ {\mathcal{T}_i^t},{{\mathbf{p}}^t} \right\}}}\left[ {\left. {\sum\limits_{i \in {\mathcal{H}^t}} {{{\left\| {{\mathbf{g}}_i^t} \right\|}^2}} } \right|{\mathcal{F}^t}} \right]\nonumber \\
      &+ 2\mathbb{E}\left( {\left. {{{\left\| {\frac{1}{H}\sum\limits_{i \in {\mathcal{H}^t}} {{\mathbf{g}}_i^t}  - {{\boldsymbol{\mu }}^t}} \right\|}^2}} \right|{\mathcal{F}^t}} \right)\nonumber\\
       = & 2\frac{1}{{{H^2}}}\delta \sum\limits_{i \in {\mathcal{H}^t}} {\mathbb{E}\left( {\left. {{{\left\| {{\mathbf{g}}_i^t} \right\|}^2}} \right|{\mathcal{F}^t}} \right)} \nonumber \\
       &+ 2\mathbb{E}\left( {\left. {{{\left\| {\frac{1}{H}\sum\limits_{i \in {\mathcal{H}^t}} {{\mathbf{g}}_i^t}  - {{\boldsymbol{\mu }}^t}} \right\|}^2}} \right|{\mathcal{F}^t}} \right),
\end{align}
where ${\left\langle 1 \right\rangle }$ is obtained from (\ref{basic ineq}) and ${\left\langle 2 \right\rangle }$ holds due to (\ref{def unbiased}) and (\ref{def unbiased bound}). 

Substituting (\ref{1st term}) into (\ref{exp bar hat}), we have
\begin{align}
    \label{exp bar hat 1}
    &\mathbb{E}\left( {\left. {\frac{1}{H}\sum\limits_{i \in {\mathcal{H}^t}} {{{\left\| {{{{\mathbf{\bar g}}}^t} - {\mathbf{\hat g}}_i^t} \right\|}^2}} } \right|{\mathcal{F}^t}} \right) \nonumber\\
    \leqslant & 4\frac{1}{{{H^2}}}\delta \sum\limits_{i \in {\mathcal{H}^t}} {\mathbb{E}\left( {\left. {{{\left\| {{\mathbf{g}}_i^t} \right\|}^2}} \right|{\mathcal{F}^t}} \right)}  + 4\mathbb{E}\left( {\left. {{{\left\| {\frac{1}{H}\sum\limits_{i \in {\mathcal{H}^t}} {{\mathbf{g}}_i^t}  - {{\boldsymbol{\mu }}^t}} \right\|}^2}} \right|{\mathcal{F}^t}} \right)\nonumber \\
       & + 4\mathbb{E}\left( {\left. {\frac{1}{H}\sum\limits_{i \in {\mathcal{H}^t}} {{{\left\| {{\mathbf{\hat g}}_i^t - {\mathbf{g}}_i^t} \right\|}^2}} } \right|{\mathcal{F}^t}} \right)+ 4\mathbb{E}\left( {\left. {{{\left\| {{\mathbf{g}}_i^t - {{\boldsymbol{\mu }}^t}} \right\|}^2}} \right|{\mathcal{F}^t}} \right)\nonumber\\
    \leqslant & \left( {\frac{{4\delta }}{{{H^2}}} + \frac{{4\delta }}{H}} \right)\sum\limits_{i \in {\mathcal{H}^t}} {\mathbb{E}\left( {\left. {{{\left\| {{\mathbf{g}}_i^t} \right\|}^2}} \right|{\mathcal{F}^t}} \right)}\nonumber \\
       &  + 4\mathbb{E}\left( {\left. {{{\left\| {\frac{1}{H}\sum\limits_{i \in {\mathcal{H}^t}} {{\mathbf{g}}_i^t}  - {{\boldsymbol{\mu }}^t}} \right\|}^2}} \right|{\mathcal{F}^t}} \right) + 4\mathbb{E}\left( {\left. {{{\left\| {{\mathbf{g}}_i^t - {{\boldsymbol{\mu }}^t}} \right\|}^2}} \right|{\mathcal{F}^t}} \right).
\end{align}
In (\ref{exp bar hat 1}), by defining 
\begin{align}
    \label{def G}
    {{\mathbf{G}}^t} \triangleq {\left[ {\nabla {f_{p_1^t}}\left( {\mathbf{x}} \right),...,\nabla {f_{p_N^t}}\left( {\mathbf{x}} \right)} \right]^T},
\end{align}
we can express
\begin{align}
    \label{1st term 1}
  &  \mathbb{E}\left( {\left. {{{\left\| {\frac{1}{H}\sum\limits_{i \in {\mathcal{H}^t}} {{\mathbf{g}}_i^t}  - {{\boldsymbol{\mu }}^t}} \right\|}^2}} \right|{\mathcal{F}^t}} \right) \nonumber \\
       =& \mathbb{E}\left( {\left. {{{\left\| {\frac{1}{d}\frac{1}{H}{{\mathbf{h}}^{1 \times N}}{{\mathbf{\hat{S}}}}{{\mathbf{G}}^t} - \frac{1}{N}{{\mathbf{1}}^{1 \times N}}{{\mathbf{G}}^t}} \right\|}^2}} \right|{\mathcal{F}^t}} \right)\nonumber\\
     \leqslant & \mathbb{E}\left( {\left. {{{\left\| {\left( {\frac{1}{d}\frac{1}{H}{{\mathbf{h}}^{1 \times N}}{{\mathbf{\hat{S}}}} - \frac{1}{N}{{\mathbf{1}}^{1 \times N}}} \right)} \right\|}^2}{{\left\| {{{\mathbf{G}}^t}} \right\|}^2_F}} \right|{\mathcal{F}^t}} \right)\nonumber\\
     =&\left\| {{{\mathbf{G}}^t}} \right\|_F^2\mathbb{E}\left( {\left. {{{\left\| {\left( {\frac{1}{d}\frac{1}{H}{{\mathbf{h}}^{1 \times N}}{\mathbf{\hat S}} - \frac{1}{N}{{\mathbf{1}}^{1 \times N}}} \right)} \right\|}^2}} \right|{\mathcal{F}^t}} \right),
\end{align}
where we have
\begin{align}
    \label{G norm}
  &\left\| {{{\mathbf{G}}^t}} \right\|_F^2 = \sum\limits_k {{{\left\| {\nabla {f_k}\left( {{{\mathbf{x}}^t}} \right)} \right\|}^2}} \nonumber \\
        = & \sum\limits_k {{{\left\| {\nabla {f_k}\left( {{{\mathbf{x}}^t}} \right) - \frac{1}{N}\nabla F\left( {{{\mathbf{x}}^t}} \right) + \frac{1}{N}\nabla F\left( {{{\mathbf{x}}^t}} \right)} \right\|}^2}}  \nonumber \\
  = &\sum\limits_k {{{\left\| {\nabla {f_k}\left( {{{\mathbf{x}}^t}} \right) - \frac{1}{N}\nabla F\left( {{{\mathbf{x}}^t}} \right)} \right\|}^2}}  + \sum\limits_k {{{\left\| {\frac{1}{N}\nabla F\left( {{{\mathbf{x}}^t}} \right)} \right\|}^2}}  \nonumber \\
  \mathop  \leqslant \limits^{\left\langle 1 \right\rangle } & N{\beta ^2} + \frac{1}{N}{\left\| {\nabla F\left( {{{\mathbf{x}}^t}} \right)} \right\|^2},
\end{align}
with $\left\langle 1 \right\rangle$ derived from Assumption~\ref{assp bounded heter}. 

Substituting (\ref{S bound}) and (\ref{G norm}) into (\ref{1st term 1}) yields 
\begin{align}
    \label{1st term 2}
& \mathbb{E}\left( {\left. {{{\left\| {\frac{1}{H}\sum\limits_{i \in {\mathcal{H}^t}} {{\mathbf{g}}_i^t}  - {{\boldsymbol{\mu }}^t}} \right\|}^2}} \right|{\mathcal{F}^t}} \right) \nonumber \\
       \leqslant & N{\beta ^2}\frac{{\left( {N - H} \right)\left( {N - d} \right)}}{{dH\left( {N - 1} \right)N}}\nonumber \\
       & + \frac{1}{N}{\left\| {\nabla F\left( {{{\mathbf{x}}^t}} \right)} \right\|^2}\frac{{\left( {N - H} \right)\left( {N - d} \right)}}{{dH\left( {N - 1} \right)N}}.
\end{align}

Substituting (\ref{h1}) and (\ref{1st term 2}) into (\ref{exp bar hat 1}), we have
\begin{align}
    \label{1st term 3}
  &  \mathbb{E}\left( {\left. {\frac{1}{H}\sum\limits_{i \in {\mathcal{H}^t}} {{{\left\| {{{{\mathbf{\bar g}}}^t} - {\mathbf{\hat g}}_i^t} \right\|}^2}} } \right|{\mathcal{F}^t}} \right)\nonumber\\
    \leqslant & \left( {\frac{{4\delta }}{{{H^2}}} + \frac{{4\delta }}{H}} \right)\sum\limits_{i \in {\mathcal{H}^t}} {\mathbb{E}\left( {\left. {{{\left\| {{\mathbf{g}}_i^t} \right\|}^2}} \right|{\mathcal{F}^t}} \right)} \nonumber \\
       & + 4N{\beta ^2}\frac{{\left( {N - H} \right)\left( {N - d} \right)}}{{dH\left( {N - 1} \right)N}}  \nonumber \\
       &+ 4\frac{1}{N}{\left\| {\nabla F\left( {{{\mathbf{x}}^t}} \right)} \right\|^2}\frac{{\left( {N - H} \right)\left( {N - d} \right)}}{{dH\left( {N - 1} \right)N}} \nonumber\\
    &+ 4\frac{{\left( {N - d} \right){\beta ^2}}}{{d\left( {N - 1} \right)}}.
\end{align}
In (\ref{1st term 3}), we have
\begin{align}
    \label{1st term 4}
    &\sum\limits_{i \in {\mathcal{H}^t}} {\mathbb{E}\left( {\left. {{{\left\| {{\mathbf{g}}_i^t} \right\|}^2}} \right|{\mathcal{F}^t}} \right)}  = \frac{H}{{N{d^2}}}\left\| {{{\mathbf{\hat{S}}}}{{\mathbf{G}}^t}} \right\|_F^2 \leqslant \frac{H}{{N{d^2}}}\left\| {{{\mathbf{\hat{S}}}}} \right\|_F^2\left\| {{{\mathbf{G}}^t}} \right\|_F^2 \nonumber \\
       &= \frac{H}{d}\left( {N{\beta ^2} + \frac{1}{N}{{\left\| {\nabla F\left( {{{\mathbf{x}}^t}} \right)} \right\|}^2}} \right),
\end{align}
where the last inequality holds due to (\ref{G norm}). 
Substituting (\ref{1st term 4}) into (\ref{1st term 3}), we have
\begin{align}
    \label{1st term 5}
  &  \mathbb{E}\left( {\left. {\frac{1}{H}\sum\limits_{i \in {\mathcal{H}^t}} {{{\left\| {{{{\mathbf{\bar g}}}^t} - {\mathbf{\hat g}}_i^t} \right\|}^2}} } \right|{\mathcal{F}^t}} \right) \nonumber \\
        \leqslant & \left[ {\left( {\frac{1}{H} + 1} \right)\frac{{4\delta }}{d} + 4\frac{{\left( {N - H} \right)\left( {N - d} \right)}}{{dH\left( {N - 1} \right)N}}} \right] \nonumber\\
      & \times \left( {N{\beta ^2} + \frac{1}{N}{{\left\| {\nabla F\left( {{{\mathbf{x}}^t}} \right)} \right\|}^2}} \right) \nonumber \\
       &+ 4\frac{{\left( {N - d} \right){\beta ^2}}}{{d\left( {N - 1} \right)}}.
\end{align}
Based on Definition~\ref{def:resaveraging}, we have
\begin{align}
    \label{agg bound}
    \mathbb{E}\left( {\left. {{{\left\| {{{{\mathbf{\hat g}}}^t} - {{{\mathbf{\bar g}}}^t}} \right\|}^2}} \right|{\mathcal{F}^t}} \right) \leqslant \kappa\mathbb{E}\left( {\left. {\frac{1}{H}\sum\limits_{i \in {\mathcal{H}^t}} {{{\left\| {{{{\mathbf{\bar g}}}^t} - {\mathbf{\hat g}}_i^t} \right\|}^2}} } \right|{\mathcal{F}^t}} \right).
\end{align}
Combining (\ref{1st term 5}) and (\ref{agg bound}) completes the proof. 

\section{Proof of Corollary~\ref{coro}}
\label{appen coro1}
From the derivation of (\ref{1st term 2}), according to Lemma~\ref{lemma1}, the infimum of the upper bound of $\mathbb{E}\left( {\left. {{{\left\| {\frac{1}{H}\sum\limits_{i \in {\mathcal{H}^t}} {{\mathbf{g}}_i^t}  - {{\boldsymbol{\mu }}^t}} \right\|}^2}} \right|{\mathcal{F}^t}} \right)$ can be attained when constructing the computation task matrix as $\mathbf{\hat{S}}$. Based on that, from (\ref{exp bar hat 1}) and (\ref{agg bound}), Corollary~\ref{coro} can be proved easily. 
\section{Proof of Lemma~\ref{lemma4}}
\label{appen lemma4}
We can easily derive that 
\begin{align}
    \label{proof 4 1}
  &  \mathbb{E}\left( {\left. {{{\left\| {{{{\mathbf{\bar g}}}^t}} \right\|}^2}} \right|{\mathcal{F}^t}} \right) = \mathbb{E}\left( {\left. {{{\left\| {{{{\mathbf{\bar g}}}^t} - {{\boldsymbol{\mu }}^t} + {{\boldsymbol{\mu }}^t}} \right\|}^2}} \right|{\mathcal{F}^t}} \right)   \nonumber \\ \mathop        \leqslant  \limits^{\left\langle 1 \right\rangle } & 2\mathbb{E}\left( {\left. {{{\left\| {{{{\mathbf{\bar g}}}^t} - {{\boldsymbol{\mu }}^t}} \right\|}^2}} \right|{\mathcal{F}^t}} \right) + 2\mathbb{E}\left( {\left. {{{\left\| {{{\boldsymbol{\mu }}^t}} \right\|}^2}} \right|{\mathcal{F}^t}} \right)\nonumber\\
      \mathop  \leqslant  \limits^{\left\langle 2 \right\rangle } & 4\frac{1}{{{H^2}}}\delta \sum\limits_{i \in {\mathcal{H}^t}} {\mathbb{E}\left( {\left. {{{\left\| {{\mathbf{g}}_i^t} \right\|}^2}} \right|{\mathcal{F}^t}} \right)} \nonumber\\
      &+ 4\mathbb{E}\left( {\left. {{{\left\| {\frac{1}{H}\sum\limits_{i \in {\mathcal{H}^t}} {{\mathbf{g}}_i^t}  - {{\boldsymbol{\mu }}^t}} \right\|}^2}} \right|{\mathcal{F}^t}} \right) + 2\mathbb{E}\left( {\left. {{{\left\| {{{\boldsymbol{\mu }}^t}} \right\|}^2}} \right|{\mathcal{F}^t}} \right)\nonumber\\
       \mathop  \leqslant  \limits^{\left\langle 3 \right\rangle }  &\left[ {\frac{{4\delta }}{{Hd}} + 4\frac{{\left( {N - H} \right)\left( {N - d} \right)}}{{dH\left( {N - 1} \right)N}}} \right]\left( {N{\beta ^2} + \frac{1}{N}{{\left\| {\nabla F\left( {{{\mathbf{x}}^t}} \right)} \right\|}^2}} \right)\nonumber\\
       &+ 2\frac{1}{{{N^2}}}{\left\| {\nabla F\left( {{{\mathbf{x}}^t}} \right)} \right\|^2},
\end{align}
where $\left\langle 1 \right\rangle$ holds due to (\ref{basic ineq}), $\left\langle 2 \right\rangle$ is obtained by substituting (\ref{1st term}) into (\ref{proof 4 1}), and $\left\langle 3 \right\rangle$ is derived by substituting (\ref{1st term 2}) and (\ref{1st term 4}) into (\ref{proof 4 1}). 
\section{Proof of Theorem~\ref{convergence performance}}
\label{appen proof the}
Based on (\ref{honest avg}), we have
\begin{align}
    \label{exp honest avg}
 & \mathbb{E}\left[ {\left. {{{{\mathbf{\bar g}}}^t}} \right|{\mathcal{F}^t}} \right] = \frac{1}{H}\sum\limits_{i \in {\mathcal{H}^t}} {\mathbb{E}\left[ {\left. {{\mathbf{\hat g}}_i^t} \right|{\mathcal{F}^t}} \right]}  = \frac{1}{H}\sum\limits_{i \in {\mathcal{H}^t}} {\mathbb{E}\left[ {\left. {{\mathbf{g}}_i^t} \right|{\mathcal{F}^t}} \right]} \nonumber\\
  &= \frac{1}{H}\sum\limits_{i \in {\mathcal{H}^t}} {\mathbb{E}\left[ {\left. {\sum\limits_{k \in \left\{ {\left. k \right|{\hat{s}}\left( {\mathcal{T}_i^t,k} \right) \ne 0} \right\}} {\frac{1}{d}} \nabla {f_k}\left( {{{\mathbf{x}}^t}} \right)} \right|{\mathcal{F}^t}} \right]} \nonumber\\
  =& \frac{1}{H}\sum\limits_{i \in {\mathcal{H}^t}} {\mathbb{E}\left[ {\left. {\frac{1}{N}\nabla F\left( {{{\mathbf{x}}^t}} \right)} \right|{\mathcal{F}^t}} \right]}  
   = \boldsymbol{\mu}^t,
\end{align}
where the second equality holds due to (\ref{def unbiased}) and (\ref{def unbiased bound}), the third equality holds according to (\ref{encoding}), and the fourth equality can be derived from the randomness of the task indices received by the devices and the construction of the computation task matrix. 
Based on (\ref{exp honest avg}) and Assumption~\ref{smooth assumption}, we can derive 
\begin{align}
    \label{smooth proof}
  &\mathbb{E}\left[ {\left. {F\left( {{{\mathbf{x}}^{t + 1}}} \right)} \right|{\mathcal{F}^t}} \right] \nonumber \\
   \leqslant& F\left( {{{\mathbf{x}}^t}} \right) + \mathbb{E}\left[ {\left. {\left\langle {\nabla F\left( {{{\mathbf{x}}^t}} \right),{{\mathbf{x}}^{t + 1}} - {{\mathbf{x}}^t}} \right\rangle } \right|{\mathcal{F}^t}} \right] \nonumber \\
   &+ \frac{L}{2}\mathbb{E}\left( {\left. {{{\left\| {{{\mathbf{x}}^{t + 1}} - {{\mathbf{x}}^t}} \right\|}^2}} \right|{\mathcal{F}^t}} \right) \nonumber \\
   \mathop  = \limits^{\left\langle 1 \right\rangle } & F\left( {{{\mathbf{x}}^t}} \right) - {\gamma ^t}\mathbb{E}\left[ {\left. {\left\langle {\nabla F\left( {{{\mathbf{x}}^t}} \right),{{{\mathbf{\hat g}}}^t}} \right\rangle } \right|{\mathcal{F}^t}} \right]+ \frac{L}{2}\mathbb{E}\left( {\left. {{{\left\| {{\gamma ^t}{{{\mathbf{\hat g}}}^t}} \right\|}^2}} \right|{\mathcal{F}^t}} \right) \nonumber \\
   =& F\left( {{{\mathbf{x}}^t}} \right) - {\gamma ^t}\mathbb{E}\left[ {\left. {\left\langle {\nabla F\left( {{{\mathbf{x}}^t}} \right),{{{\mathbf{\hat g}}}^t} - {{{\mathbf{\bar g}}}^t} + {{{\mathbf{\bar g}}}^t}} \right\rangle } \right|{\mathcal{F}^t}} \right] \nonumber \\ 
   &+ \frac{{L{{\left( {{\gamma ^t}} \right)}^2}}}{2}\mathbb{E}\left( {\left. {{{\left\| {{{{\mathbf{\hat g}}}^t} - {{{\mathbf{\bar g}}}^t} + {{{\mathbf{\bar g}}}^t}} \right\|}^2}} \right|{\mathcal{F}^t}} \right) \nonumber \\
   \mathop  \leqslant \limits^{\left\langle 2 \right\rangle }& F\left( {{{\mathbf{x}}^t}} \right) - {\gamma ^t}\mathbb{E}\left[ {\left. {\left\langle {\nabla F\left( {{{\mathbf{x}}^t}} \right),{{{\mathbf{\hat g}}}^t} - {{{\mathbf{\bar g}}}^t}} \right\rangle } \right|{\mathcal{F}^t}} \right] \nonumber \\
   &- {\gamma ^t}\mathbb{E}\left[ {\left. {\left\langle {\nabla F\left( {{{\mathbf{x}}^t}} \right),{{{\mathbf{\bar g}}}^t}} \right\rangle } \right|{\mathcal{F}^t}} \right] \nonumber \\
   &+ L{\left( {{\gamma ^t}} \right)^2}\mathbb{E}\left( {\left. {{{\left\| {{{{\mathbf{\hat g}}}^t} - {{{\mathbf{\bar g}}}^t}} \right\|}^2}} \right|{\mathcal{F}^t}} \right) + L{\left( {{\gamma ^t}} \right)^2}\mathbb{E}\left( {\left. {{{\left\| {{{{\mathbf{\bar g}}}^t}} \right\|}^2}} \right|{\mathcal{F}^t}} \right) \nonumber \\
  \mathop  \leqslant \limits^{\left\langle 3 \right\rangle }& F\left( {{{\mathbf{x}}^t}} \right) + {\gamma ^t}\frac{{ {\eta{\left\| {\nabla F\left( {{{\mathbf{x}}^t}} \right)} \right\|}^2} + \eta^{-1}\mathbb{E}\left[ {\left. {{{\left\| {{{{\mathbf{\hat g}}}^t} - {{{\mathbf{\bar g}}}^t}} \right\|}^2}} \right|{\mathcal{F}^t}} \right]}}{2}\nonumber \\
  &- {\gamma ^t}\frac{1}{N}{\left\| {\nabla F\left( {{{\mathbf{x}}^t}} \right)} \right\|^2} \nonumber \\
  & + L{\left( {{\gamma ^t}} \right)^2}\mathbb{E}\left( {\left. {{{\left\| {{{{\mathbf{\hat g}}}^t} - {{{\mathbf{\bar g}}}^t}} \right\|}^2}} \right|{\mathcal{F}^t}} \right) + L{\left( {{\gamma ^t}} \right)^2}\mathbb{E}\left( {\left. {{{\left\| {{{{\mathbf{\bar g}}}^t}} \right\|}^2}} \right|{\mathcal{F}^t}} \right),
\end{align}
$\forall \eta>0$, where $\left\langle 1 \right\rangle$ can be derived by substituting (\ref{update global model}) into (\ref{smooth proof}), $\left\langle 2 \right\rangle$ holds due to (\ref{basic ineq}), and $\left\langle 3 \right\rangle$ is obtained by applying Young's Inequality together with (\ref{exp honest avg}). 

Substituting (\ref{bound honest avg}) and (\ref{g bar}) into (\ref{smooth proof}), we have
\begin{align}
    \label{smooth proof1}
   & \mathbb{E}\left[ {\left. {F\left( {{{\mathbf{x}}^{t + 1}}} \right)} \right|{\mathcal{F}^t}} \right] \nonumber\\
    \leqslant &  F\left( {{{\mathbf{x}}^t}} \right) + {\gamma ^t}\frac{{\eta {{\left\| {\nabla F\left( {{{\mathbf{x}}^t}} \right)} \right\|}^2} + {\eta ^{ - 1}}\kappa {\kappa _1} + {\eta ^{ - 1}}\kappa {\kappa _2}{{\left\| {\nabla F\left( {{{\mathbf{x}}^t}} \right)} \right\|}^2}}}{2} \nonumber\\
    &- {\gamma ^t}\frac{1}{N}{\left\| {\nabla F\left( {{{\mathbf{x}}^t}} \right)} \right\|^2} \nonumber \\
 &  + L{\left( {{\gamma ^t}} \right)^2}\left[ {\kappa {\kappa _1} + \kappa {\kappa _2}{{\left\| {\nabla F\left( {{{\mathbf{x}}^t}} \right)} \right\|}^2}} \right]\nonumber\\
 &+ L{\left( {{\gamma ^t}} \right)^2}\left( {{\kappa _3} + {\kappa _4}{{\left\| {\nabla F\left( {{{\mathbf{x}}^t}} \right)} \right\|}^2}} \right).
\end{align}
Taking full expectations on both sides of (\ref{smooth proof1}) and rearranging the terms, we obtain
\begin{align}
    \label{smooth proof1 exp}
 & \mathbb{E}\left[ {F\left( {{{\mathbf{x}}^{t + 1}}} \right)} \right] \nonumber\\
 \leqslant & \mathbb{E}\left[ {F\left( {{{\mathbf{x}}^t}} \right)} \right]+ {\gamma ^t}\frac{{{\eta ^{ - 1}}\kappa {\kappa _1}}}{2} + {\left( {{\gamma ^t}} \right)^2}\left( {L\kappa {\kappa _1} + L{\kappa _3}} \right)\nonumber\\
& + \mathbb{E}\left[ {{{\left\| {\nabla F\left( {{{\mathbf{x}}^t}} \right)} \right\|}^2}} \right] \nonumber\\
&\times \left[ {{\gamma ^t}\left( {\frac{{\eta  + {\eta ^{ - 1}}\kappa {\kappa _2}}}{2} - \frac{1}{N}} \right) + {{\left( {{\gamma ^t}} \right)}^2}\left( {L\kappa {\kappa _2} + L{\kappa _4}} \right)} \right].
\end{align}
Let us take $\eta  = \sqrt {\kappa{\kappa _2}}$. In this case, (\ref{smooth proof1 exp}) can be rewritten as
\begin{align}
    \label{smooth exp2}
  &\mathbb{E}\left[ {{{\left\| {\nabla F\left( {{{\mathbf{x}}^t}} \right)} \right\|}^2}} \right]\left[ {{\gamma ^t}\left( {\frac{1}{N} - \sqrt {\kappa {\kappa _2}} } \right) - {{\left( {{\gamma ^t}} \right)}^2}\left( {L\kappa {\kappa _2} + L{\kappa _4}} \right)} \right] \nonumber \\
  & \leqslant \mathbb{E}\left[ {F\left( {{{\mathbf{x}}^t}} \right)} \right] - \mathbb{E}\left[ {F\left( {{{\mathbf{x}}^{t + 1}}} \right)} \right] + {\gamma ^t}\frac{{\kappa {\kappa _1}}}{{2\sqrt {\kappa {\kappa _2}} }} \nonumber\\
  &+ {\left( {{\gamma ^t}} \right)^2}\left( {L\kappa {\kappa _1} + L{\kappa _3}} \right).
\end{align}
If we fix the learning rate as ${\gamma ^t} = {\gamma ^0}  < \frac{{\frac{1}{N} - \sqrt {\kappa {\kappa _2}} }}{{L\kappa {\kappa _2} + L{\kappa _4}}}$, which is small enough so that ${{\gamma ^0}\left( {\frac{1}{N} - \sqrt {\kappa {\kappa _2}} } \right) - {{\left( {{\gamma ^0}} \right)}^2}\left( {L\kappa {\kappa _2} + L{\kappa _4}} \right)}>0$, under the condition $\sqrt {\kappa {\kappa _2}}  < \frac{1}{N}$, we can rewrite (\ref{smooth exp2}) as
\begin{align}
    \label{smooth exp3}
 & \mathbb{E}\left[ {{{\left\| {\nabla F\left( {{{\mathbf{x}}^t}} \right)} \right\|}^2}} \right] \nonumber\\
 \leqslant & \frac{{\mathbb{E}\left[ {F\left( {{{\mathbf{x}}^t}} \right)} \right] - \mathbb{E}\left[ {F\left( {{{\mathbf{x}}^{t + 1}}} \right)} \right]}}{{{\gamma ^0}\left( {\frac{1}{N} - \sqrt {\kappa {\kappa _2}} } \right) - {{\left( {{\gamma ^0}} \right)}^2}\left( {L\kappa {\kappa _2} + L{\kappa _4}} \right)}} \nonumber\\
 &+ \frac{{{\gamma ^0}\frac{{\kappa {\kappa _1}}}{{2\sqrt {\kappa {\kappa _2}} }} + {{\left( {{\gamma ^0}} \right)}^2}\left( {L\kappa {\kappa _1} + L{\kappa _3}} \right)}}{{{\gamma ^0}\left( {\frac{1}{N} - \sqrt {\kappa {\kappa _2}} } \right) - {{\left( {{\gamma ^0}} \right)}^2}\left( {L\kappa {\kappa _2} + L{\kappa _4}} \right)}}.
\end{align}
Averaging over $T$ iterations of (\ref{smooth exp3}), we can derive (\ref{converg cola}) based on Assumption~\ref{lower bound}, which completes the proof. 
   
\end{document}